\documentclass[12pt,reqno]{article}

\usepackage[usenames]{xcolor}
\usepackage{amssymb}
\usepackage{amsmath}
\usepackage{amsthm}
\usepackage{amsfonts}
\usepackage{amscd}
\usepackage{graphicx}
\usepackage{lscape}
\usepackage[title]{appendix}

\usepackage[colorlinks=true,
linkcolor=webgreen,
filecolor=webbrown,
citecolor=webgreen]{hyperref}

\definecolor{webgreen}{rgb}{0,.5,0}
\definecolor{webbrown}{rgb}{.6,0,0}

\usepackage{fullpage}
\usepackage{float}

\usepackage{graphics}
\usepackage{latexsym}
\usepackage{epsf}

\usepackage{wrapfig}

\newcommand{\seqnum}[1]{\href{https://oeis.org/#1}{\rm \underline{#1}}}
\def\suchthat{\, : \, }

\newcommand{\inn}[1]{\langle #1 \rangle}
\usepackage[noend]{algpseudocode}
\usepackage{algorithm}
\usepackage{tikz}
\usepackage[skip=3pt]{caption}
\usepackage{mathtools}

\begin{document}

\theoremstyle{plain}
\newtheorem{theorem}{Theorem}
\newtheorem{corollary}[theorem]{Corollary}
\newtheorem{lemma}[theorem]{Lemma}
\newtheorem{proposition}[theorem]{Proposition}

\theoremstyle{definition}
\newtheorem{definition}[theorem]{Definition}
\newtheorem{example}[theorem]{Example}
\newtheorem{conjecture}[theorem]{Conjecture}
\newtheorem{openproblem}[theorem]{Open Problem}

\theoremstyle{remark}
\newtheorem{remark}[theorem]{Remark}

\title{Transduction of Automatic Sequences and Applications}

\author{
Jeffrey Shallit\footnote{Research funded by a grant from NSERC, 2018-04118.}\\
School of Computer Science\\
University of Waterloo\\
Waterloo, ON  N2L 3G1\\
Canada\\
\href{mailto:shallit@uwaterloo.ca}{\tt shallit@uwaterloo.ca}\\
\ \\
Anatoly Zavyalov\\
University of Toronto\\
Toronto, ON M5S 2E4 \\
Canada\\
\href{mailto:anatoly.zavyalov@mail.utoronto.ca}{\tt anatoly.zavyalov@mail.utoronto.ca}
}

\maketitle

\begin{abstract}
We consider the implementation of the transduction of automatic sequences, and their generalizations, in the {\tt Walnut} software for solving decision problems in combinatorics on words.   We provide a number of applications, including (a) representations of $n!$ as a sum of three squares (b) overlap-free Dyck words and (c) sums of Fibonacci representations. We also prove results about iterated running sums of the Thue-Morse sequence.
\end{abstract}

\section{Introduction}
The $k$-automatic sequences form an interesting class that has been studied for more than fifty years now \cite{Cobham:1972}.   One nice property is that the first-order theory of these sequences, with addition, is algorithmically decidable
\cite{Bruyere&Hansel&Michaux&Villemaire:1994}.   \texttt{Walnut\/}, a free software system created and implemented by Hamoon Mousavi \cite{Mousavi:2016}, makes it possible to state and evaluate the truth of nontrivial first-order statements about automatic sequences, often in a matter of seconds \cite{Shallit:2022}.   With it, we can easily reprove old theorems in a simple and uniform way, explore unresolved conjectures, and prove new theorems.

This class is defined as follows:  a sequence
$(a(n))_{n \geq 0}$ over a finite alphabet is $k$-automatic if there exists a deterministic finite automaton with output (DFAO) that, on completely processing an input of $n$ expressed in base $k$, reaches a state $s$ with output $a(n)$.

According to a famous theorem of Cobham \cite{Cobham:1972}, the class of $k$-automatic sequences is closed under deterministic $t$-uniform transductions.  The particular model of transducer here is the following:  outputs are associated with transitions, so that every input letter results in the output of exactly $t$ letters.     
Cobham's theorem was later extended to the more general class of morphic sequences (defined below) by Dekking \cite{Dekking:1994}.

Transducers make it possible to manipulate automatic sequences in useful and interesting ways.  For example, it follows that the running sum and running product (taken
modulo a natural number $M \geq 2$) 
of a $k$-automatic sequence taking
values in $\Sigma_t = \{ 0,1,\ldots, t-1 \}$
is $k$-automatic.

In this paper we report on a recent implementation of Dekking's result in {\tt Walnut}, carried out by the second author.
This new capability of transducers has been implemented in the latest version of {\tt Walnut}, which is available for free download at\\ 
\centerline{\url{https://cs.uwaterloo.ca/~shallit/walnut.html} \ .}

With it we are able to obtain a number of results, including reproving the results of Burns \cite{Burns:2022} on the representation of $n!$ as a sum of three squares, in a simpler and more general way. 
We also prove a new result on overlap-free Dyck words, and several results about iterated running sums of the Thue-Morse sequence. Lastly, we discuss the transduction of Fibonacci automata and automata in other numeration systems.

\section{Preliminaries}

\subsection{Basic definitions}

Given a finite alphabet $\Sigma$, a \textit{word} over $\Sigma$ is a concatenation of elements of $\Sigma$. If $w = w_1 \cdots w_k$ is a finite word over $\Sigma$, then the length of $w$ is denoted by $|w| = k$. The set of all finite words over $\Sigma$ is denoted by $\Sigma^*$, which also includes the empty word of length $0$, denoted by $\epsilon$. The \textit{reversal} of $w = w_1 \cdots w_k$ is defined as $w^R = w_k \cdots w_1$. An infinite word ${\bf a} = a_0 a_1 a_2 \cdots$ over $\Sigma$ is a sequence $(a_n)_{n \geq 0}$ of elements of $\Sigma$. 

Let $Q$ be a finite set. A function $h \colon Q^* \to Q^*$ is called a \textit{morphism} if $h(vw) = h(v) h(w)$ for all $v, w \in Q^*$. Note that a morphism must, by definition, satisfy $h(\epsilon) = \epsilon$. Furthermore, we say that a morphism $h$ is \textit{$k$-uniform} if $|h(q)| = k$ for every $q \in Q$. A morphism can be applied to infinite words and sequences in the obvious way: if ${\bf q} = q_0 q_1 q_2 \cdots$ is an infinite word over $Q$, then $h({\bf q}) = h(q_0) h(q_1) h(q_2) \cdots$.

If $h$ is a morphism, by $h^n$ we mean the $n$-fold
composition of $h$ with itself.
If $h$ is prolongable, that is, if
$h(a) = ax$ and $h^n(x) \not= \epsilon$ for all $n \geq 0$, then we can define the infinite word $h^{\omega}(a) := a x h(x) h^2(x) \cdots$; it is a
fixed point of $h$.

Let $\Delta$ be a finite alphabet. A \textit{coding} $\lambda \colon Q^* \to \Delta^*$ is a $1$-uniform morphism (satisfying $|\lambda(q)| = 1$ for all $q \in Q$). An infinite sequence $\bf x$ over $\Delta$ is a \textit{morphic sequence\/} if ${\bf x} = \lambda(h^\omega(a))$ where $\lambda$ is a coding, $h$ is a prolongable morphism, and $a$ is a letter.

\subsection{Numeration systems} \label{sec:numeration}

We give a brief overview of the numeration systems covered in this paper. See  \cite[Chapter 3]{Allouche&Shallit:2003} for an in-depth treatment of numeration systems.

For $k = 2, 3, 4, \ldots$, the \textit{most-significant-digit-first (msd) base-$k$ representation} of $n \in \mathbb{N}$ is $(n)_k = d_t d_{t-1} \cdots d_1 d_0$ where $n = \sum_{i=0}^t d_i k^i$ and $d_i \in \{0, \ldots, k-1\}$ for all $i = 0, \ldots, t$. We abbreviate this as the \textit{msd}-$k$ representation of $n$. The least-significant-digit-first (lsd) base-$k$ representation of $n$ is $d_0 d_1 \cdots d_{t-1} d_t$. For example, the msd-$2$ representation of $13$ is $1101$, and the lsd-$2$ representation of $13$ is $1011$.

Define $F_0 = 0$, $F_1 = 1$, and $F_n = F_{n-1} + F_{n-2}$ for $n \geq 2$. The most-significant-digit-first \textit{Fibonacci representation} of an integer $n \geq 0$ is $(n)_{F} = d_t d_{t-1} \cdots d_1 d_0$, where $n = \sum_{i=0}^t d_i F_{i+2}$ and $d_i \in \{0, 1\}$, with no consecutive $1$s, i.e., there does not exist $0 \leq i \leq t-1$ such that $d_i = d_{i+1} = 1$. Every integer $n \geq 0$ can be uniquely written in this way \cite{Lekkerkerker:1952,Zeckendorf:1972}. For example, the most-significant-digit-first Fibonacci representation of $14$ is $(14)_{F} = 100001$. If $n = \sum_{i=0}^t d_i F_{i+2}$, then the least-significant-digit-first Fibonacci representation of $n$ is $d_0 d_1 \cdots d_{t-1} d_t$. 

\subsection{DFAOs}

A \textit{deterministic finite automaton with output} (DFAO) is a $6$-tuple $M = \inn{Q, \Sigma, \delta, q_0, \Delta, \lambda}$, where 
\begin{itemize}
    \item $Q$ is a finite set of \textit{states};
\item $\Sigma$ is a finite set representing the \textit{input alphabet};
\item $\delta \colon Q \times \Sigma \to Q$ is the \textit{transition function};
\item $q_0 \in Q$ is the \textit{initial state};
\item $\Delta$ is a finite set representing the \textit{output alphabet}; and 
\item $\lambda \colon Q^* \to \Delta^*$ is the \textit{coding}, defined by a single value $\lambda(q)$ for each state $q \in Q$.
\end{itemize}
We can extend the transition function to words over $Q$ in the obvious way.

If $\Sigma = \{0, \ldots,  k-1 \}$ for some $k \in \mathbb{N}$, we consider the sequence $(x_n)_{n \geq 0}$ \textit{computed} by $M$, defined as
$$
x_n = \lambda(\delta(q_0, (n)_k)),
$$
where $(n)_k$ denotes the msd-$k$ representation of $n$. Intuitively, $x_n$ is the output when $M$ receives the msd-$k$ representation of $n$ as input. 

The sequence ${\bf x} = (x_n)_{n \geq 0}$ is called \textit{$k$-automatic} if there exists a DFAO $M$ with input alphabet $\Sigma = \{0, \ldots, k-1\}$ that computes it. Theorem 3 of Cobham \cite{Cobham:1972} proves that a $k$-automatic sequence $\bf x$ is a morphic sequence with respect to some $k$-uniform morphism. We provide the construction given in Cobham below, as it will be useful to us in Section \ref{sec:implementation}. Given a DFAO $M$ with alphabet $\Sigma = \{0, 1, \ldots, k-1 \}$, define a morphism $h \colon Q^* \to Q^*$ by
$$
h(a) = \delta(a, 0) \; \delta(a, 1) \; \cdots \; \delta(a, k-1).
$$
Then the sequences $\bf x$ and $\lambda(h^{\omega}(q_0))$ are equal. Conversely, given a $k$-automatic sequence $\mathbf{x} = (x_n)_{n \geq 0} = \lambda(h^{\omega}(q_0))$ where $q_0 \in Q$, a $k$-uniform morphism $h \colon Q^* \to Q^*$ and a coding $\lambda \colon Q^* \to \Delta^*$, we define the DFAO $N = \inn{Q, \Sigma, \delta, q_0, \Delta, \lambda}$ where $\Sigma = \{0, \ldots, k-1 \}$ and $\delta \colon Q \times \Sigma \to Q$ is defined by $\delta(q, i) = h(q)[i]$, where $h(q) = h(q)[0] \, \cdots \, h(q)[k-1]$. Then $N$ computes the sequence $\bf x$.

\subsection{Transducers} \label{sec:transducers}

A \textit{1-uniform deterministic finite-state transducer} (or \textit{transducer} for short) is a $6$-tuple
$$ T = \inn{V, \Delta, \varphi, v_0, \Gamma, \sigma},$$ where $V$ is a finite set of \textit{states}, $\Delta$ is a finite set representing the \textit{input alphabet}, $\varphi \colon V \times \Delta \to V$ is the \textit{transition function}, $v_0 \in V$ is the \textit{initial state}, $\Gamma$ is a finite set representing the \textit{output alphabet}, and $\sigma \colon V \times \Delta \to \Gamma$ is the \textit{output function}. A transducer can be viewed as a DFA with an output function $\sigma$ layered on top, where a single element of $\Gamma$ is output for each transition.

We are interested in taking a sequence $ {\bf x} = (x_n)_{n \geq 0}$ and \textit{transducing} it, meaning we pass the sequence through the transducer symbol by symbol, obtaining a single symbol output with each transition.   Formally, the transduction of a sequence ${\bf x} = (x_n)_{n \geq 0}$ over the alphabet $\Delta$ by a transducer $T$ is defined to be the infinite sequence
$$
T({\bf x}) = \sigma(v_0, x_0) \; \sigma(\varphi(v_0, x_0), x_1) \; \sigma(\varphi(v_0, x_0 x_1), x_2) \; \cdots \; \sigma(\varphi(v_0, x_0 \cdots x_{n-1}), x_n) \; \cdots.
$$

In this paper we only consider $1$-uniform transducers; the more general case of $t$-uniform transducers, as mentioned previously, can easily be handled by applying the appropriate $t$-uniform morphism to the output of a $1$-uniform transducer.

\begin{example}
Allouche and Bousquet-Melou \cite{Allouche&Bousquet-Melou:1994b}
studied the ``generalized Rudin Shapiro" sequences,
which are running sums, taken modulo 2, of the paperfolding sequences.

Let us show that if you start the sequence of
unfolding instructions
$$ 1, 1, -1, 1, -1, 1, -1, 1, -1, ...,$$
in other words, $1$, followed by $(1,-1)$ repeating,
and then apply a morphism $1 \rightarrow 0$, $-1 \rightarrow 1$
to the resulting
resulting paperfolding word
$1, 1, -1, -1, 1, -1, -1, \cdots$, prepend a $0$, and
transduce this with a running sum transducer, mod $2$,
you get the classical Rudin-Shapiro sequence.

We need to define a transducer \texttt{RUNSUM2},
which transduces a sequence $a_0 a_1 a_2 \cdots $ over $\{0, 1\}$ into the sequence $b_0 b_1 b_2 \cdots$ where the $k$'th term is the running sum mod $2$ of the subsequence $a_0 a_1 a_2 \cdots a_k$, that is, $b_k = (\sum_{i=0}^k a_i) \bmod 2$. The transducer is illustrated in Figure~\ref{fig:RUNSUM2}. See the Appendix for the \texttt{Walnut} definition of the transducer.
\begin{figure}[H] 
\begin{center}
    \includegraphics[width=3in]{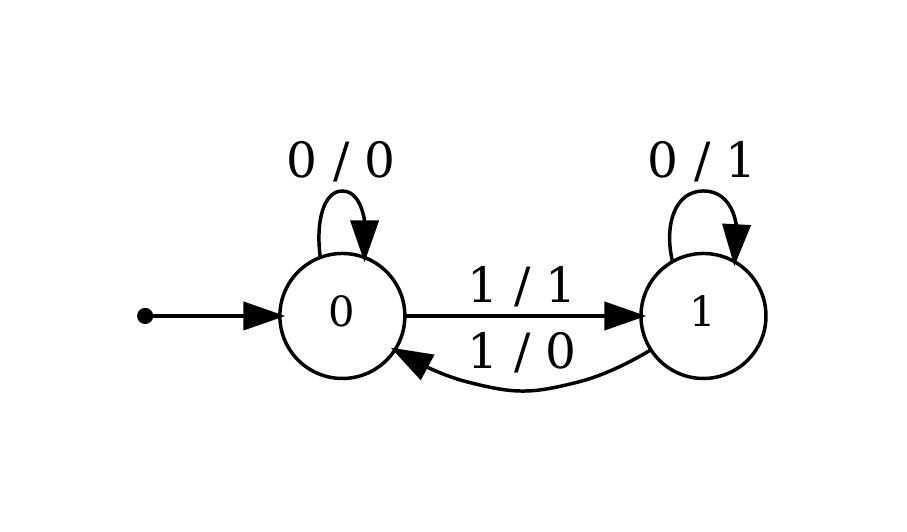}
\end{center}
\caption{Transducer \texttt{RUNSUM2} for running sum mod 2.} \label{fig:RUNSUM2}
\end{figure}

Then we use the following {\tt Walnut} code.
\begin{verbatim}
reg apfcode {-1,0,1} "1(1[-1])*0*":
# make unfolding instructions 1, 1, -1, 1, -1, 1, -1, ...

def apf "?lsd_2 Ex $apfcode(x) & FOLD[x][n]=@-1":
# return i in {0,1} if paperfolding sequence equals (-1)^i at position n

def apfm "?msd_2 `$apf(n)":
# turn an lsd-first automaton into an msd-first

combine PF1 apfm:
# turn it into a DFAO

transduce RS3 RUNSUM2 PF1:
# transduce it with running sum mod 2 transducer

eval check "An RS3[n]=RS[n]":
# check if it's the same as Rudin-Shapiro.  It is!
\end{verbatim}
\label{exam1}
\end{example}

\begin{example} \label{exm:tmrunsum}
    Consider the famous Thue-Morse sequence $\mathbf{t} = h^{\omega}(0)$ defined by the $2$-uniform morphism $h \colon \Sigma^* \to \Sigma^*$ given by $h(0) = 01$, $h(1) = 10$, with $\Sigma = \{0, 1\}$. The first few terms of the sequence are:
    $$
    \mathbf{t} = 0110100110010110 \cdots.
    $$
    Suppose we wish to compute the running sum (mod $2$) of the Thue-Morse sequence. We can do this using transduction: define the transducer $T = \inn{V = \{v_0, v_1\}, \Sigma, \varphi, v_0, \Sigma, \sigma}$ where the transition function $\varphi: V \times \Sigma \to V$ is defined by
    $$
    \varphi(v_0, 0) = v_0, \quad \varphi(v_0, 1) = v_1, \quad \varphi(v_1, 0) = v_1, \quad \varphi(v_1, 1) = v_0,
    $$
    and the output function $\sigma: V \times \Sigma \to \Sigma$ is defined by 
    $$
        \sigma(v_0, 0) = 0, \quad \sigma(v_0, 1) = 1, \quad
        \sigma(v_1, 0) = 1, \quad \sigma(v_1, 1) = 0.
    $$
    Then
    $$
    T(\mathbf{t}) = 0100111011100101 \cdots
    $$
    is the running sum of $\mathbf{t}$. 
    The transducer $T$ is illustrated in Figure \ref{fig:RUNSUM2}.
    The resulting sequence
    is sequence \seqnum{A255817} in the
    {\it On-Line Encyclopedia of Integer Sequences} (OEIS) \cite{Sloane:2023}.
    
    If we iteratively apply $T$ to $\mathbf{t}$ arbitrarily often to generate running sums $T(\mathbf{t}), T^2(\mathbf{t}), T^3(\mathbf{t}), \ldots$ (where $T^2(\mathbf{t}) \coloneqq T(T(\mathbf{t}), T^3(\mathbf{t}) \coloneqq T(T(T(\mathbf{t})))$, etc.), we can plot each running sum on a separate row to get a fractal consisting of Sierpi\'nski triangles, as illustrated in Figure \ref{fig:thueMorseRunSum}. This structure was also found by Prunescu \cite{Prunescu:2011}. 
    \begin{figure}[htb]
    \begin{center}
        \includegraphics[width=3.5in]{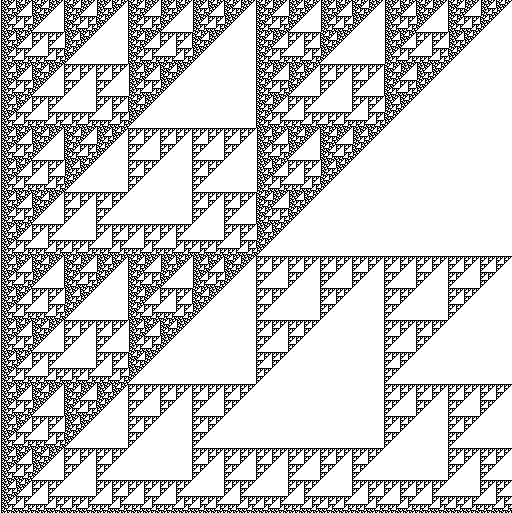}
    \end{center}
    \caption{The first 512 terms of the first 512 iterated running sums of the Thue-Morse sequence. The $k$'th row represents the terms of the sequence $T^k(\mathbf{t})$.}
    \label{fig:thueMorseRunSum}
    \end{figure}
\label{one}
\end{example}

\begin{example}
Let ${\bf p} = 10111010\cdots$ be the 
{\it period-doubling sequence} \cite{Damanik:2000},
the fixed point of the morphism $1 \rightarrow 10$,
$0 \rightarrow 11$, and sequence
\seqnum{A035263} in the OEIS.  Applying the transducer
of Example~\ref{one} to $\bf p$ gives
$\overline{\bf t} = 11010011 \cdots$, the Thue-Morse sequence with its first symbol removed.
\end{example}

\section{Dekking's result} \label{sec:dekking}

It turns out that the transduction of an $k$-automatic sequence always produces another $k$-automatic sequence, formalized in the next theorem.
\begin{theorem} \label{thm:automaticTransduce}
    Let $\mathbf{x} = (x_n)_{n \geq 0}$ be a $k$-automatic sequence over $\Delta$, and let $T = \inn{V, \Delta, \varphi, v_0, \Gamma, \sigma}$ be a transducer. Then the sequence 
    $$
    T(\mathbf{x}) = \sigma(v_0, x_0) \; \sigma(\varphi(v_0, x_0), x_1) \; \sigma(\varphi(v_0, x_0 x_1), x_2) \; \cdots \; \sigma(\varphi(v_0, x_0 \cdots x_{n-1}), x_n) \; \cdots
    $$
    is $k$-automatic.
\end{theorem}

Dekking \cite{Dekking:1994} provided a constructive argument that we use to prove this theorem, which we recapitulate for the convenience of the reader. 

We first give a few definitions. Suppose that ${\bf x} = (x_n)_{n \geq 0}$ is a $k$-automatic sequence, so $\mathbf{x} =  \lambda({\bf q})$ where $h \colon Q^* \to Q^*$ is a prolongable $k$-uniform morphism, the coding $\lambda$ maps $ Q^* \to \Delta^*$, and ${\bf q} = h^\omega(q_0) = q_0 q_1 q_2 \cdots$. 
For all $y \in \Delta^*$, define the function $f_y \colon V \to V$ by $f_y(v) = \varphi^*(v, y)$. Note that $f_{y_1 y_2} = f_{y_2} \circ f_{y_1}$ for all $y_1, y_2 \in \Delta^*$.
We need the following lemma.
\begin{lemma} \label{lem:ultimatelyperiodic}
    For all $w \in Q^*$, the sequences $\mathcal{I}_{\infty}(w) = (f_{\lambda(w)}, f_{\lambda(h(w))}, f_{\lambda(h^2(w))}, \ldots, f_{\lambda(h^n(w))}, \ldots)$  are ultimately periodic with the same period and preperiod, i.e., there exist integers $p \geq 1$, $r \geq 0$ such that $f_{\lambda(h^i(w))} = f_{\lambda(h^{p+i}(w))}$ for all $w \in Q^*$ and $i \geq r$.
\end{lemma}
\begin{proof}
    Suppose that $Q = \{q_1, \ldots, q_m\}$. Define the map $\Phi \colon \left(V^V\right)^m \to \left(V^V\right)^m$ such that for every $w_1, \ldots, w_m \in Q^*$, we have
    $$
    \Phi(f_{\lambda(w_1)}, \ldots, f_{\lambda(w_m)}) = (f_{\lambda(h(w_1))}, \ldots, f_{\lambda(h(w_m))}). 
    $$
    Then, for all $n \geq 1$, it is obvious that
    $$
    \Phi^n(f_{\lambda(w_1)}, \ldots, f_{\lambda(w_m)})) = (f_{\lambda(h^n(w_1))}, \ldots, f_{\lambda(h^n(w_m))}).
    $$
    Recall that the \textit{orbit} of $(f_{\lambda(q_1)}, \ldots, f_{\lambda(q_m)})$ under $\Phi$ is the sequence
    \begin{align*}
        &\left((f_{\lambda(q_1)}, \ldots, f_{\lambda(q_m)}), \Phi(f_{\lambda(q_1)}, \ldots, f_{\lambda(q_m)}), \ldots, \Phi^n(f_{\lambda(q_1)}, \ldots, f_{\lambda(q_m)}), \ldots\right) \\
        & \quad = \left((f_{\lambda(q_1)}, \ldots, f_{\lambda(q_m)}), (f_{\lambda(h(q_1))}, \ldots, f_{\lambda(h(q_m))}), \ldots, (f_{\lambda(h^n(q_1))}, \ldots, f_{\lambda(h^n(q_m))}), \ldots\right).
    \end{align*}
    As $\left(V^V\right)^m$ is a finite set, the orbit of $(f_{\lambda(q_1)}, \ldots, f_{\lambda(q_m)})$ under $\Phi$ is ultimately periodic with some period $p \geq 1$ and preperiod $r \geq 0$, meaning that for all $j \in \mathbb{N}$ with $j \geq r$, we have
    \begin{equation*}
        (f_{\lambda(h^j(q_1))}, \ldots, f_{\lambda(h^j(q_m))}) = (f_{\lambda(h^{p+j}(q_1))}, \ldots, f_{\lambda(h^{p+j}(q_m))}),
    \end{equation*}
    so $f_{\lambda(h^j(q_i))} = f_{\lambda(h^{p+j}(q_i))}$ for each  $i = 1, \ldots, m$, and by comparing each coordinate, we see that the sequences $\mathcal{I}_{\infty}(q) = (f_{\lambda(q)}, f_{\lambda(h(q))}, f_{\lambda(h^2(q))}, \ldots, f_{\lambda(h^n(q))}, \ldots)$ are ultimately periodic with the same period $p$ and preperiod $r$ for all $q \in Q$. Now, for $w = w_1 \cdots w_\ell \in Q^*$, we have
    \begin{align*}
        \mathcal{I}_{\infty}(w) &= (f_{\lambda(w)}, f_{\lambda(h(w))}, \ldots, f_{\lambda(h^n(w))}, \ldots) \\
        &= (f_{\lambda(w_1) \cdots \lambda(w_\ell)},
        f_{\lambda(h(w_1)) \cdots \lambda(h(w_\ell))}, \ldots, f_{\lambda(h^n(w_1)) \cdots \lambda(h^n(w_\ell))}, \ldots) \\
        &= (f_{\lambda(w_\ell)} \circ \cdots \circ f_{\lambda(w_1)},
        f_{\lambda(h(w_\ell))} \circ \cdots \circ f_{\lambda(h(w_1))},
        f_{\lambda(h^n(w_\ell))} \circ \cdots \circ f_{\lambda(h^n(w_1))}, \ldots).
    \end{align*}
    As each of the $\mathcal{I}_{\infty}(w_i)$ is ultimately periodic with period $p$ and preperiod $r$, it follows that $\mathcal{I}_{\infty}(w)$ is also ultimately periodic with period $p$ and preperiod $r$, as desired.
\end{proof}

We can now prove Theorem~\ref{thm:automaticTransduce}.
\begin{proof}[Proof of \textbf{Theorem \ref{thm:automaticTransduce}}]
By Lemma \ref{lem:ultimatelyperiodic}, the sequence $\mathcal{I}_{\infty}(w)$ is ultimately periodic with the same period $p$ and preperiod $r$ for every $w \in Q^*$. For each $w \in Q^*$, we define
$$
\mathcal{I}(w) = (f_{\lambda(w)}, f_{\lambda(h(w))}, \ldots, f_{\lambda(h^{p+r-1}(w))}).
$$
We also define
\begin{equation}
    \widetilde{Q} = \{(a, \mathcal{I}(w)) \suchthat a \in Q, w \in Q^* \}. \label{def:tildeQ}
\end{equation}
We define our $k$-uniform morphism $\widetilde{h} \colon \widetilde{Q}^* \to \widetilde{Q}^*$ by
\begin{equation}
    \widetilde{h}(a, \mathcal{I}(w)) = (h(a)_1, \mathcal{I}(h(w))) \:(h(a)_2, \mathcal{I}(h(w) h(a)_1)) \: \cdots \: (h(a)_k, \mathcal{I}(  h(w) h(a)_1 \cdots h(a)_{k-1})), \label{eq:morphism}
\end{equation}
where $h(a)_i$ denotes the $i$'th symbol of $h(a)$. Note that $(q_0, \mathcal{I}(\epsilon)) = \widetilde{q}_0$ is the first symbol of $\widetilde{h}(q_0, \mathcal{I}(\epsilon))$, as $h(q_0)$ starts with $q_0$ and $h(\epsilon) = \epsilon$. We can then consider the fixed point 
$$
\widetilde{\bf q} = \left(\widetilde{q}_n \right)_{n \geq 0} = \widetilde{h}^{\omega}(q_0, \mathcal{I}(\epsilon)).
$$
Define the coding $\widetilde{\lambda} \colon \widetilde{Q}^* \to \Gamma^*$ by
\begin{equation}
    \widetilde{\lambda}(a, \mathcal{I}(w)) = \sigma(f_{\lambda(w)}(v_0), \lambda(a)) = \sigma(\varphi^*(v_0, \lambda(w)), \lambda(a)). \label{def:tildeLambda}
\end{equation}
We claim that $\widetilde{\lambda}(\widetilde{\bf q}) = T(\mathbf{x})$, from which the theorem statement follows. To show this, it is sufficient to show that for all $n \in \mathbb{N}$, we have
\begin{equation}
    \widetilde{q}_n = (q_n, \mathcal{I}(q_0 \cdots q_{n-1})). \label{eq:elementEquality}
\end{equation}

We proceed inductively. For $0 \leq n \leq k-1$, this follows from (\ref{eq:morphism}) and the fact that $x_0 = \lambda(q_0)$. 
Now, for arbitrary $m \in \mathbb{N}$, suppose that (\ref{eq:elementEquality}) holds for all $n \in \mathbb{N}$ with $0 \leq n < |h^m(q_0)|$. We show that (\ref{eq:elementEquality}) holds for $0 \leq n < |h^{m+1}(q_0)|$. Since $h$ is $k$-uniform with $k > 1$, we have $|h^{m+1}(q_0)| > |h^m(q_0)|$, so proving this inductive step shows that (\ref{eq:elementEquality}) holds for all $n \in \mathbb{N}$. As $h^{m+1}(q_0) = h(h^m(q_0))$, for each $n$ such that $|h^{m}(q_0)| \leq n < |h^{m+1}(q_0)|$, there are integers $i, j$ with $1 \leq i \leq |h^m(q_0)|$ and $1 \leq j \leq |h(q_i)| = k$ such that
\begin{equation}
    q_0 q_1 \cdots q_n = h(q_0 q_1 \cdots q_{i-1}) \, h(q_i)_1 \cdots h(q_i)_j, \label{eq:subSequence}
\end{equation}
Using $\widetilde{h}(\widetilde{\bf q}) = \widetilde{\bf q}$ and $|h(a)| = |\widetilde{h}(a, \mathcal{I}(w))| = k$ for  \(a \in Q\) and  \(w \in Q^*\), we get
\begin{equation}
    \widetilde{q}_0 \widetilde{q}_1 \cdots \widetilde{q}_n = 
\widetilde{h}(\widetilde{q}_0 \widetilde{q}_1 \cdots \widetilde{q}_{i-1}) \, \widetilde{h}(\widetilde{q}_i)_1 \cdots \widetilde{h}(\widetilde{q}_i)_j. \label{eq:tildeSubSequence}
\end{equation}
It then follows that
\begin{align*}
    \widetilde{q}_n &= \widetilde{h}(\widetilde{q}_i)_j \tag*{(looking at the last symbol of both sides in \eqref{eq:tildeSubSequence})} \\
    &= \widetilde{h}(q_i, \mathcal{I}(q_0 \cdots q_{i-1}))_j \tag*{(by the induction hypothesis and \eqref{eq:elementEquality})} \\
    &= (h(q_i)_j, \mathcal{I}(h(q_0 q_1 \cdots q_{i-1}) \: h(q_i)_1 \cdots h(q_i)_{j-1})) \tag*{(by the definition of $\widetilde{h}$)} \\
    &= (q_n, \mathcal{I}(q_0 q_1 \cdots q_{n-1})), \tag*{(by \eqref{eq:subSequence})}
\end{align*}
which completes the proof.
\end{proof}

The above proof for Theorem \ref{thm:automaticTransduce} gives an obvious construction for the DFAO $M'$ that generates $T(x)$ given a transducer $T$. Namely, given $M = \inn{Q, \Sigma, \delta, q_0, \Delta, \lambda}$ and $T = \inn{V, \Delta, \varphi, v_0, \Gamma, \sigma}$, we define $M' = \inn{\Omega, \Sigma, \widetilde{\delta}, \widetilde{q}_0, \Gamma, \widetilde{\lambda}}$ where $\Omega \subset \widetilde{Q}$ is the set of reachable states in $\widetilde{Q}$ from $\widetilde{q}_0$ (with $\widetilde{q}_0 \in \Omega$), $\widetilde{Q}$ and $\widetilde{\lambda}$ are defined as in (\ref{def:tildeQ}) and (\ref{def:tildeLambda}) respectively, and $\widetilde{\delta} \colon \widetilde{Q} \times \Sigma \to \widetilde{Q}$ is defined by
\begin{equation}
    \widetilde{\delta}((a, \mathcal{I}(w)), d) = (h(a)_d, \mathcal{I}(h(w) h(a)_1 \cdots h(a)_{d-1})).
\end{equation}

\section{Implementation details}\label{sec:implementation}

We now give a pseudocode implementation of the construction of the previous section. Following from the proof of Lemma \ref{lem:ultimatelyperiodic}, Algorithm \ref{alg:findpr} finds the period $p$ and preperiod $r$ used for the states of the automaton by computing the orbit of $(f_{\lambda(q_1)}, \ldots, f_{\lambda(q_m)})$ under $\Phi$ until a cycle is detected.
\begin{algorithm}[H] 
        \caption{Find period $p$ and preperiod $r$.} \label{alg:findpr}
        \begin{algorithmic}[H]
        \State Let $\textsc{orbitMap}$ be a new hash table
        \For {$n = 1, 2, 3, \ldots$}
            \State $\textsc{orbitEl} \gets (f_{\lambda(h^n(q_1))}, \ldots, f_{\lambda(h^n(q_{|Q|}))})$ 
            \If {$\textsc{orbitEl} \notin \textsc{orbitsMap}.keys$}
                \State \textsc{orbitsMap}.put(\textsc{orbit}, $n$)
            \Else
                \State $\ell \gets \textsc{orbitsMap}[\textsc{orbit}]$
                \State $p \gets n-\ell$
                \State $r \gets \ell$
                \State \Return $(p, r)$ 
            \EndIf
        \EndFor
        \end{algorithmic}
    \end{algorithm}

To generate the resulting automaton, we start from the initial state and perform a breadth-first search to generate only the reachable states of the resulting automaton. In the pseudocode in Algorithm \ref{alg:genauto}, we use methods $\textsc{Set-Output}(s, o)$, which sets the output of a state $s$ to $o$, and $\textsc{Add-Transtion}(s_1, s_2, i)$, which adds a transition from state $s_1$ to state $s_2$ on input $i$. 

\begin{algorithm} 
    \caption{Generate the resulting automaton.} \label{alg:genauto}
    \begin{algorithmic} 
    \State Let $\textsc{states}$ be a new empty set
    \State Let $\textsc{statesQueue}$ be a new queue
    \State $\textsc{initState} \gets (q_0, \mathcal{I}(\epsilon))$
    \State $\textsc{states}$.add(\textsc{initState})
    \State \textsc{statesQueue}.append(\textsc{initState}) 
    \While {$\textsc{statesQueue}.length>0$}
        \State $\textsc{currState} \gets (a, \mathcal{I}(w)) = \textsc{statesQueue}$.pop()
        \State $\textsc{Set-Output}(\textsc{currState}, \; \sigma(f_w(v_0), a))$
        \For {$i = 0, \ldots, k-1$}
            \State $\textsc{newState} \gets (h(a)_{i+1}, \mathcal{I}(  h(w) h(a)_1 \cdots h(a)_{i}))$
            \State $\textsc{Add-Transition}(\textsc{currState}, \; \textsc{newState}, \; i)$
            \If {\textsc{newState} $\notin$ \textsc{states}}
                \State \textsc{states}.add(\textsc{newState}) 
                \State \textsc{statesQueue}.append(\textsc{newState}) 
            \EndIf
        \EndFor
    \EndWhile
    
    \end{algorithmic}
\end{algorithm}

\subsection{Complexity analysis}

The proof of Lemma \ref{lem:ultimatelyperiodic} finds $p$ and $r$ by iterating through the finite set $(V^V)^{|Q|}$ until a repetition is found, which requires to iterate through at most $|V|^{|Q| \cdot |V|}$ maps. Consequently we get that $p, r \leq |V|^{|Q| \cdot |V|}$. The total number of states $|\Omega|$ is, in the worst case, at most $|\widetilde{Q}| \leq |Q| \cdot |V|^{(p+r) |V|} \leq |Q| \cdot |V|^{2 \cdot |V|^{|Q| \cdot |V| + 1}}$. To find the reachable states $\Omega$, we run a breadth-first search starting from the initial state, which, in the worst-case, takes time $\mathcal{O}(|Q| \cdot |V|^{(p+r) |V|} + k \cdot |Q| \cdot |V|^{(p+r) |V|}) = \mathcal{O}((k+1) \cdot |Q| \cdot |V|^{2 \cdot |V|^{|Q| \cdot |V| + 1}})$.

In practice, however, we do not usually see this kind of astonishingly large worst-case complexity.

\section{Applications}

\subsection{Factorial as a sum of three squares} \label{sec:fac_as_sum_three_sq}

Let $S_3 = \{ 0,1,2,3,4,5,6,8,\ldots \}$
be the set of natural numbers that can be
expressed as the sum of three squares of natural
numbers.
A well-known theorem of Legendre (see, for example, \cite[\S 1.5]{Nathanson:1996}) states that $N \in S_3$ if and only if it cannot be expressed in the
form $4^i (8j+7)$, where $i, j \geq 0$ are natural numbers.    It follows that the characteristic
sequence of the set $S_3$ is $2$-automatic,
a fact first observed by Cobham \cite[Example 8, p.~172]{Cobham:1972}.  Indeed, $N \in S_3$ if and only if $(N)_2$ is accepted by the automaton depicted in Figure~\ref{fig:threesq}.
\begin{figure}[H]
\begin{center}
    \includegraphics[width=6in]{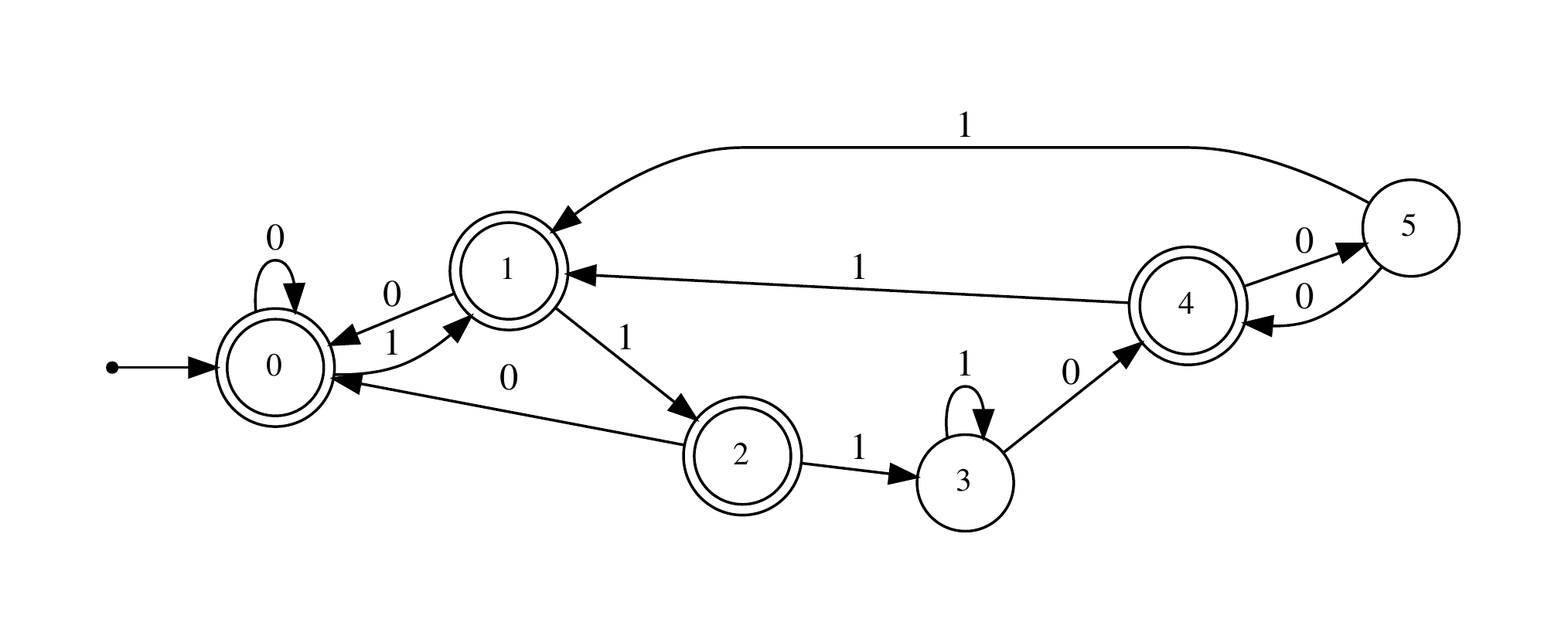}
\end{center}
\caption{Automaton accepting $(S_3)_2$.}
\label{fig:threesq}
\end{figure}

Recently a number of authors \cite{Deshouillers&Luca:2010,Hajdu&Papp:2018,Burns:2021,Burns:2022} have been interested in studying the properties of the set
$S = \{ n \suchthat n! \in S_3 \}$.  In particular, the set $S$ is $2$-automatic.   In this section, we show how to determine a DFAO for (the characteristic sequence) of $S$ by using transducers.    Let $\nu_2 (n)$ denote the exponent of the highest power of $2$ dividing $n$.   The basic idea is to use the
fact that $n! = \prod_{1 \leq i \leq n} i$, and keep track of both the parity of the exponent of the highest power of $2$ dividing $i$ and the last $3$ bits of $i/2^{\nu_2(i)}$.

Clearly
$$\nu_2(mn) = \nu_2(m) + \nu_2(n)$$
for $m,n \geq 1$.  It now follows that $\nu_2(n!) \bmod 2$ is the running sum (mod $2$) of the sequence
$\nu_2(1) \cdots \nu_2(n)$.

Let $g(n) = n/2^{\nu_2 (n)}$ for $n\geq 1$ and set $g(0)=1$.   Then it is easy to
see that $g(n) \in \{1, 3, 5, 7\}$ for all $n \geq 1$ and
$$ g(mn) \equiv g(m) g(n) \pmod 8 $$
for $m,n \geq 1$. 
 Thus
$g(n!) \bmod 8$ is the running product
(mod $8$) of the sequence
$g(1) g(2) \cdots g(n)$.  

Hence $\nu_2(n!) \bmod 2$ can be computed by a running-sum transducer, and $g(n!) \bmod 8$ can be computed by a running-product transducer.
Now $n \notin S_3$ if and only if $\nu_2(n) \equiv 0 \pmod 2$ and $g(n) \equiv 7 \pmod 8$, and therefore $n \notin S$ if and only if $\nu_2(n!) = \sum_{i = 1}^n \nu_2(i) \equiv 0 \pmod 2$ and $g(n!) \equiv \prod_{i=1}^n g(i) \equiv 7 \pmod 8$.

We can now implement these ideas in
{\tt Walnut}.
We first define the DFAO \texttt{NU\_MOD2}, which generates the sequence $(\nu_2(n) \bmod{2})_{n \geq 1}$, illustrated below:
\begin{figure}[H]
\begin{center}
    \includegraphics[width=4in]{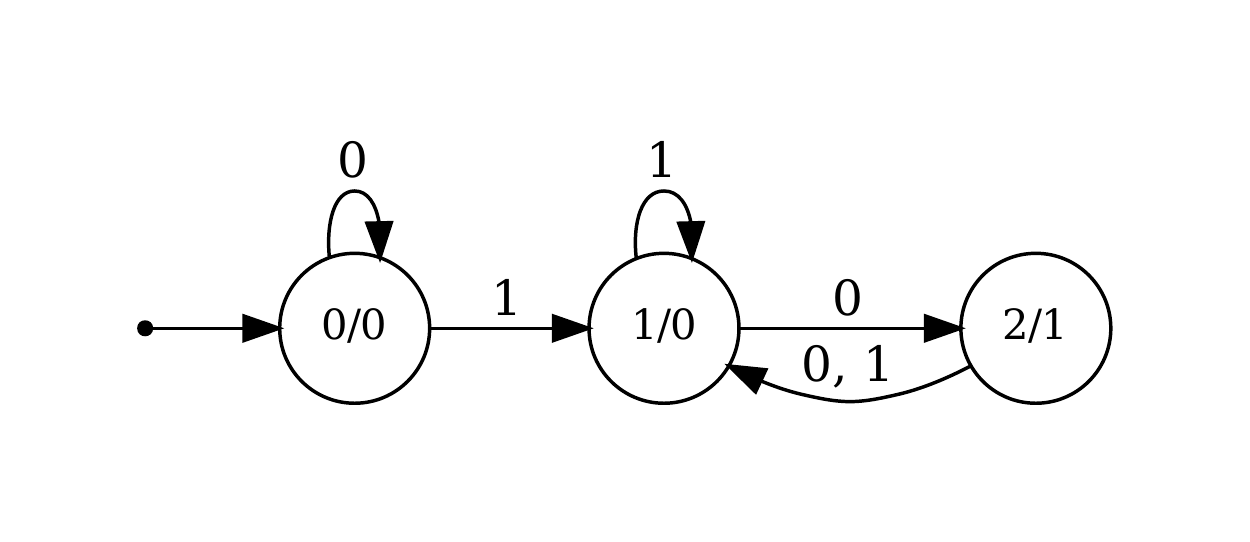}
\end{center}
\caption{DFAO \texttt{NU\_MOD2} computing $\nu_2(n) \bmod 2$.}
\end{figure}
This can be done with the following {\tt Walnut} command, which defines it via a regular expression:
\begin{verbatim}
reg nu2odd msd_2 "(0|1)*10(00)*":
combine NU_MOD2 nu2odd:
\end{verbatim}

It is easy to see that $g$ satisfies
the following identities:
\begin{align*}
g(2n) &= g(n) \\
g(4n+2) &= g(2n+1) \\
g(8n+i) &= i \quad \text{for $i \in \{1,3,5,7\}$.}
\end{align*}
From these identities, an lsd-first automaton for $g$ is trivial to derive, as illustrated in Figure~\ref{g88}.
\begin{figure}[htb]
\begin{center}
\includegraphics[width=3in]{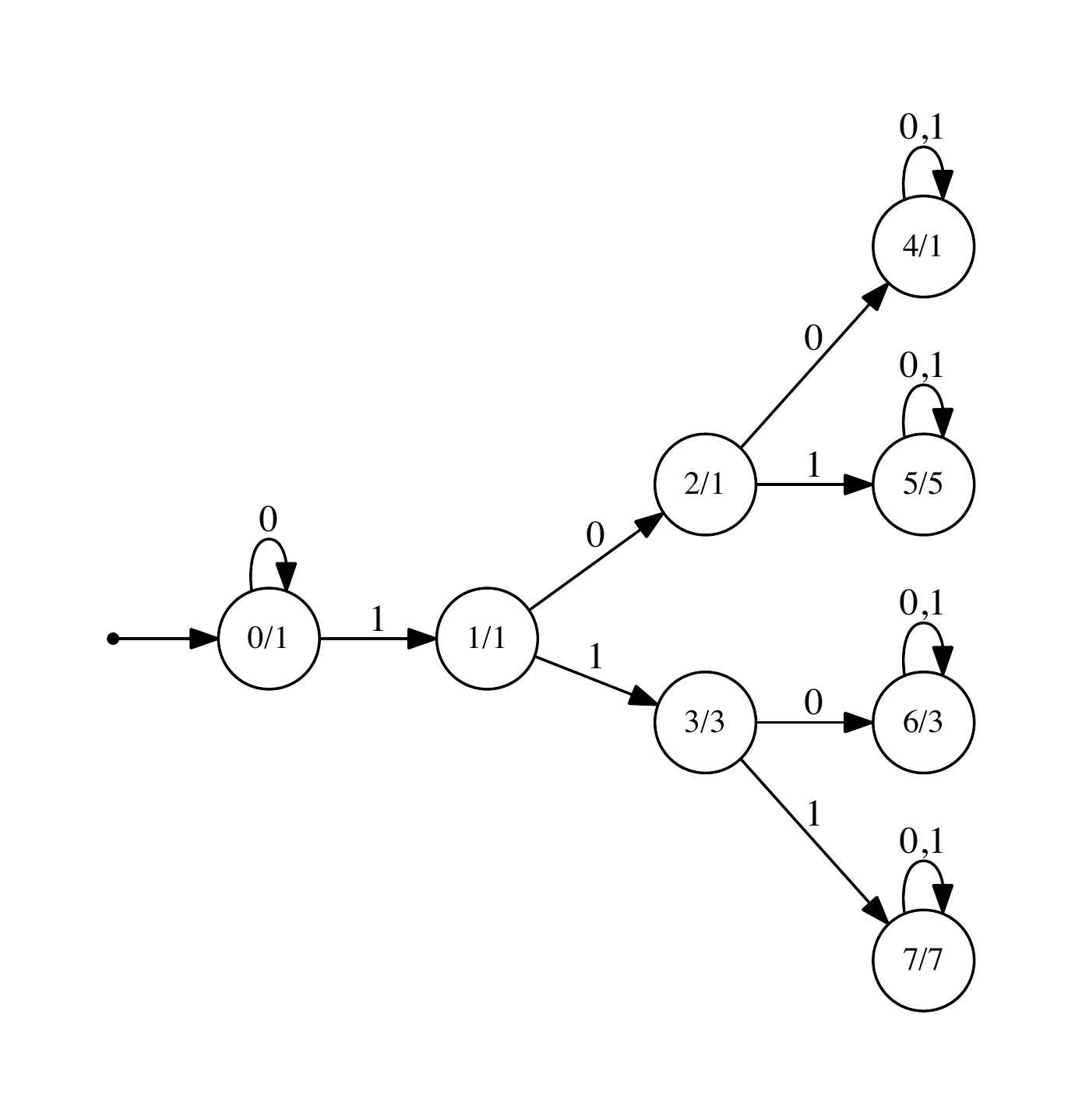}
\end{center}
\caption{DFAO computing $g(n)$ in lsd-first format.}
\label{g88}
\end{figure}
We can then reverse this automaton,
using the following {\tt Walnut} commands:
\begin{verbatim}
reverse G_MOD8 G8:
\end{verbatim}
and get
the $12$-state DFAO \texttt{G\_MOD8} that computes $g(n)$ in msd-first format.  It is illustrated in Figure~\ref{nfac}.
\begin{figure}[H]
\begin{center}
    \includegraphics[width=6in]{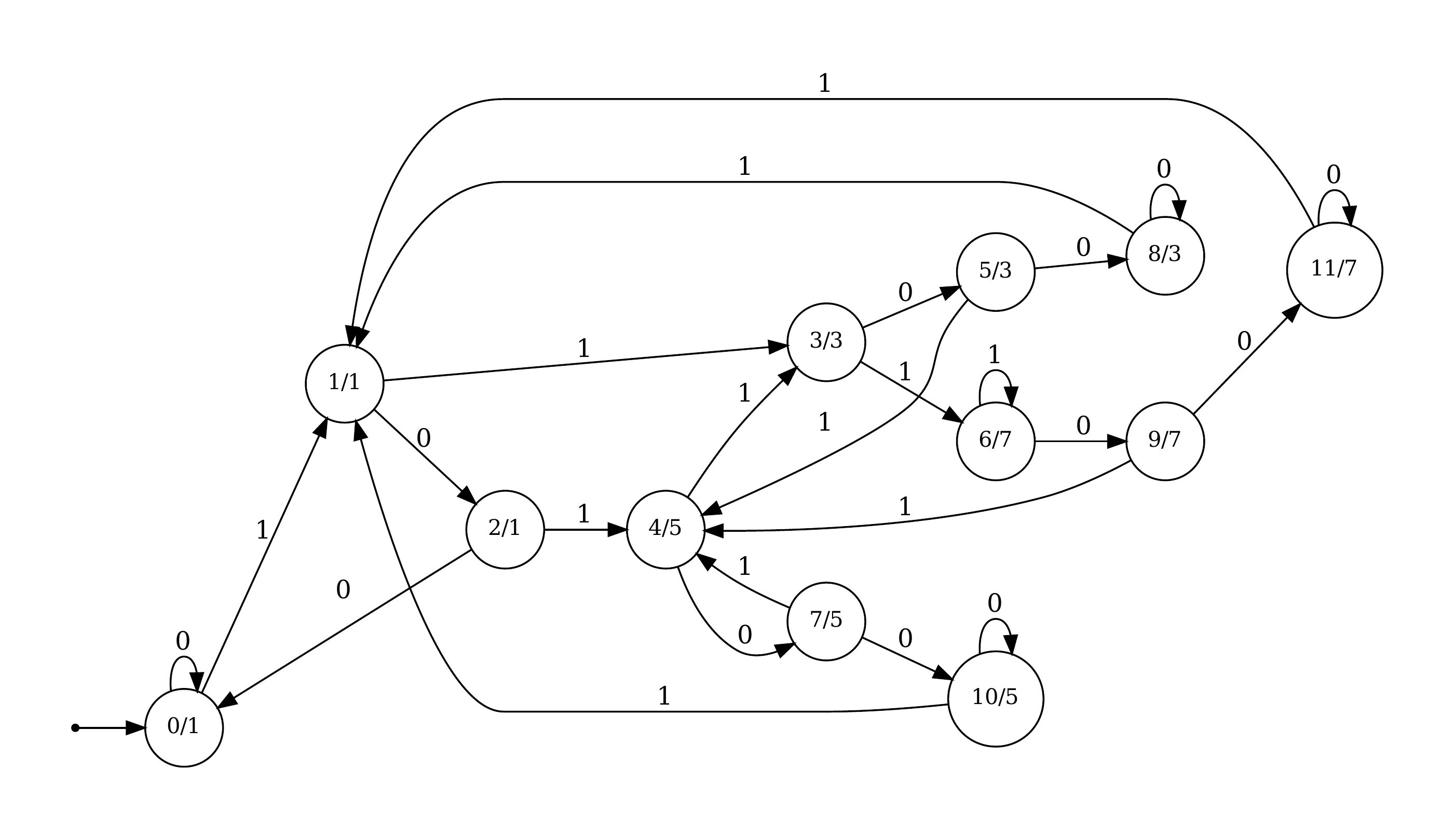}
\end{center}
\caption{DFAO \texttt{G\_MOD8} computing $g(n)$.}
\label{nfac}
\end{figure}

Lastly, we define the transducer \texttt{RUNPROD1357}, which transduces the sequence $g(1) g(2) \cdots g(n)$ into the sequence $h_1 h_2 \cdots h_n$ where the $k$'th term is the running product (mod 8) of the sequence $g(1) g(2) \cdots g(n)$, i.e., $h_k = \prod_{i=1}^k g(i) \bmod 8$. The transducer is illustrated in Figure~\ref{nfac3}, and the \texttt{Walnut} definition is given in the Appendix.
\begin{figure}[H]
\begin{center}
    \includegraphics[width=6in]{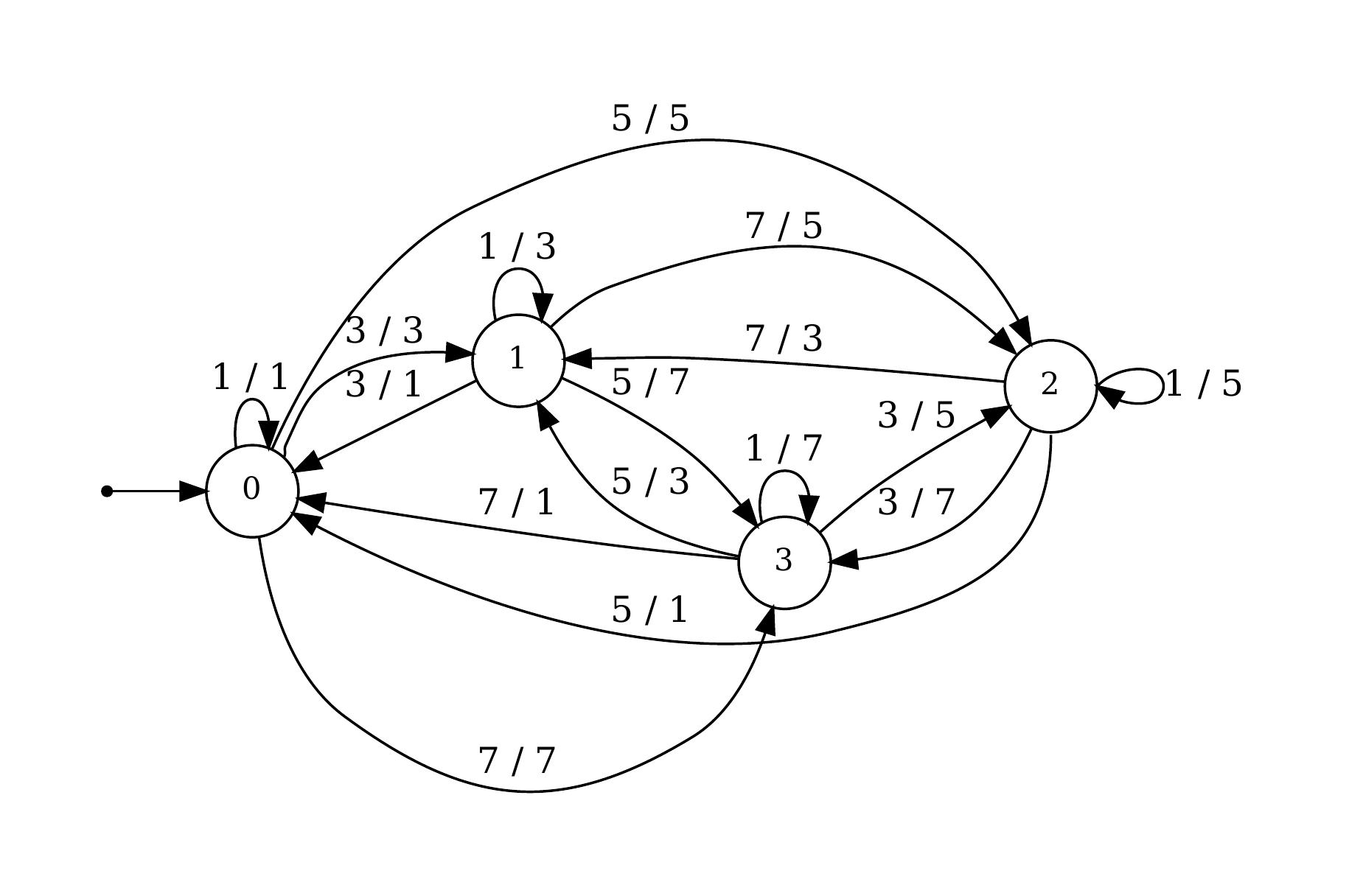}
\end{center}
\caption{Transducer \texttt{RUNPROD1357} for running product mod 8.}
\label{nfac3}
\end{figure}

We can then transduce \texttt{NU\_MOD2} with \texttt{RUNSUM2} to get a resulting DFAO \texttt{NU\_RUNSUM} using the {\tt Walnut}
command
\begin{verbatim}
transduce NU_RUNSUM RUNSUM2 NU_MOD2:
\end{verbatim}
as depicted in Figure~\ref{nfac4}.
\begin{figure}[H]
\begin{center}
    \includegraphics[width=4in]{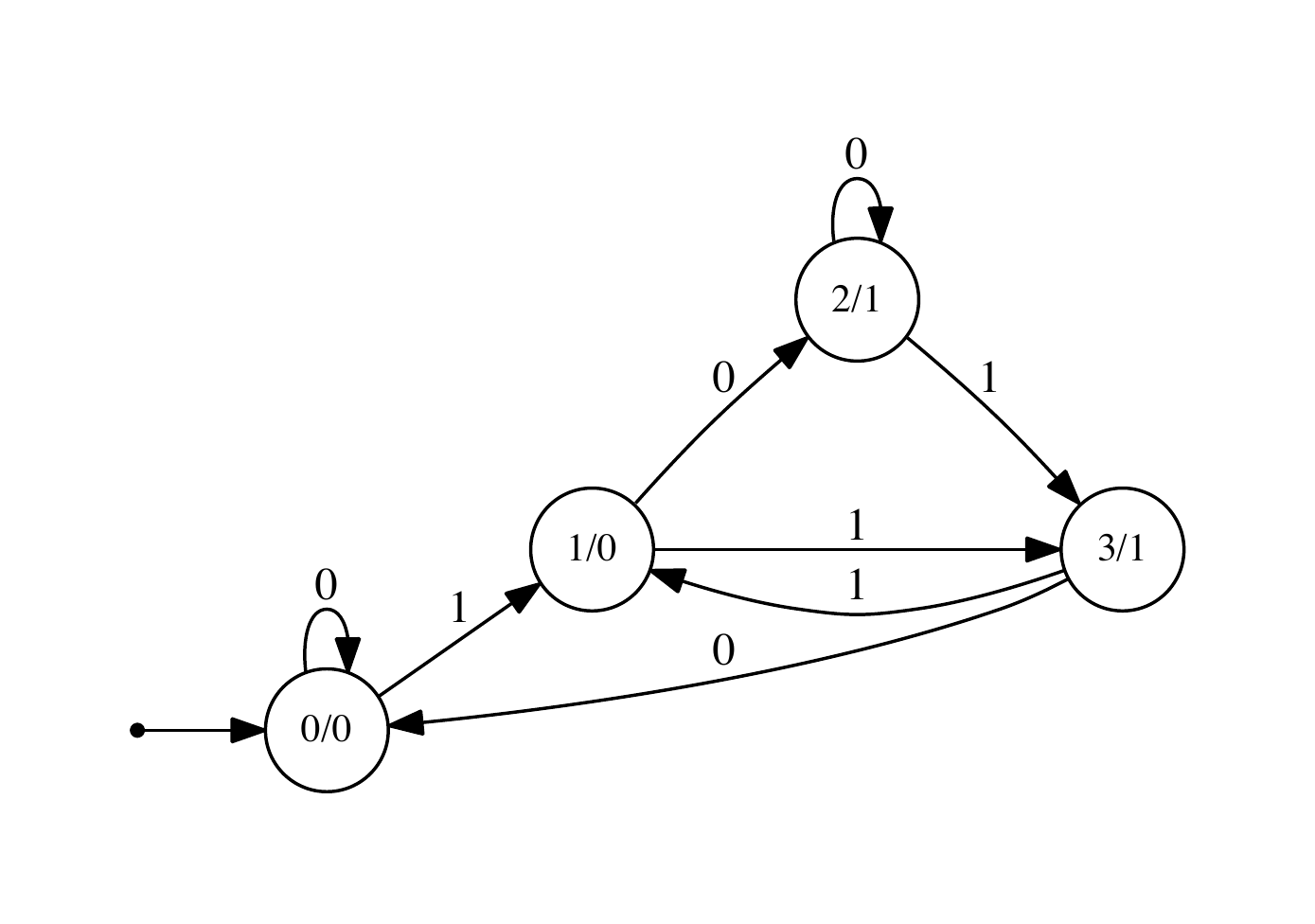}
\end{center}
\caption{DFAO \texttt{NU\_RUNSUM} computing the running sum (mod 2) of \texttt{NU\_MOD2}.}
\label{nfac4}
\end{figure}

We similarly transduce \texttt{G\_MOD8} with \texttt{RUNPROD1357} to get a resulting DFAO \texttt{G\_RUNPROD} (depicted in Figure \ref{fig:G_RUNPROD}) with the following \texttt{Walnut} command:

\begin{verbatim}
transduce G_RUNPROD RUNPROD1357 G_MOD8:
\end{verbatim}
\begin{figure}[H]
\begin{center}
    \includegraphics[width=6.5in]{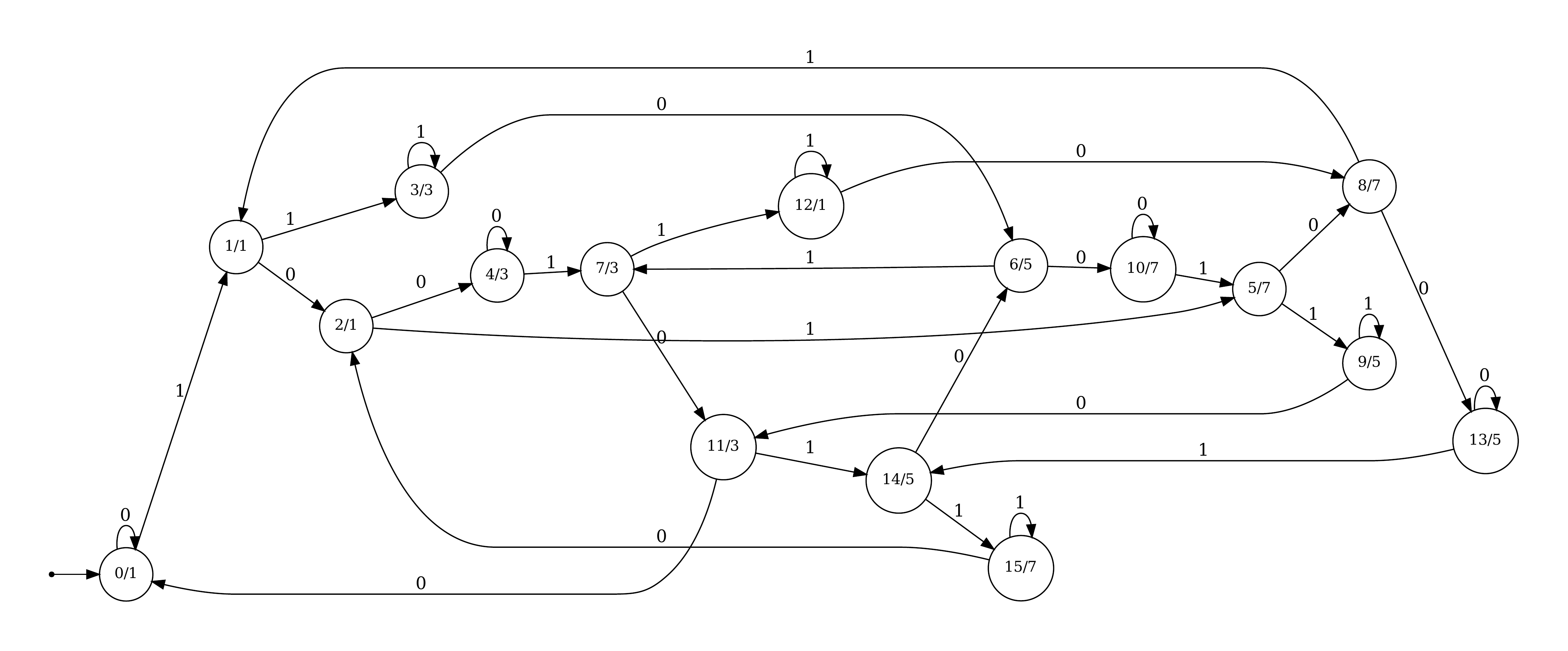}
\end{center}
\caption{DFAO \texttt{G\_RUNPROD} computing the running product (mod 8) of \texttt{G\_MOD8}.} \label{fig:G_RUNPROD}
\end{figure}

Lastly, we generate the final automaton that accepts $S$, using the characterization above that $n \in S$ if and only if $\nu_2(n!) = \sum_{i = 1}^n \nu_2(i) \equiv 1 \pmod 2$ or $g(n!) \equiv \prod_{i=1}^n g(i) \not\equiv 7 \pmod 8$, which can be directly translated into the following \texttt{Walnut} command:
\begin{verbatim}
def nfac_in_s "(NU_RUNSUM[i] = @1) | ~(G_RUNPROD[i] = @7)":
\end{verbatim}

This gives the $32$-state automaton in Figure~\ref{nfac5}.
\begin{figure}[H]
\begin{center}
    \includegraphics[width=6.5in]{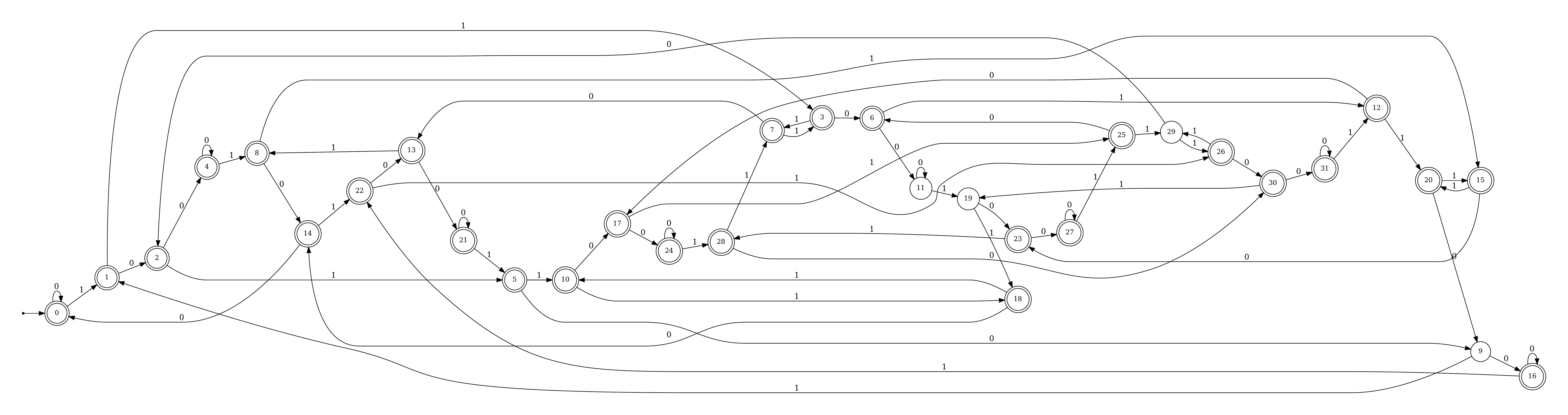}
\end{center}
\caption{Automaton accepting $S = \{ n \suchthat n! \in S_3 \}$.}
\label{nfac5}
\end{figure}

Figure 4 of Burns \cite{Burns:2022}, depicts an automaton accepting $n$, represented in base-$2$ with {\it least-significant digit\/} first, if $n!$ \textit{cannot} be written as a sum of 3 squares represented in the {\tt lsd\_k}. We can now obtain their automaton by reversing and negating the automaton \texttt{nfac\_in\_s} using the following \texttt{Walnut} command:
\begin{verbatim}
def nfac_in_s_rev_neg "~`$nfac_in_s(i)":
\end{verbatim}

This gives us the $35$-state automaton in Figure~\ref{nfac6}, which is
identical to Figure 4 of Burns \cite{Burns:2022}.   We have therefore rederived Burns' result in a simpler and more natural way.
\begin{figure}[H]
\begin{center}
   \quad \includegraphics[width=6.5in]{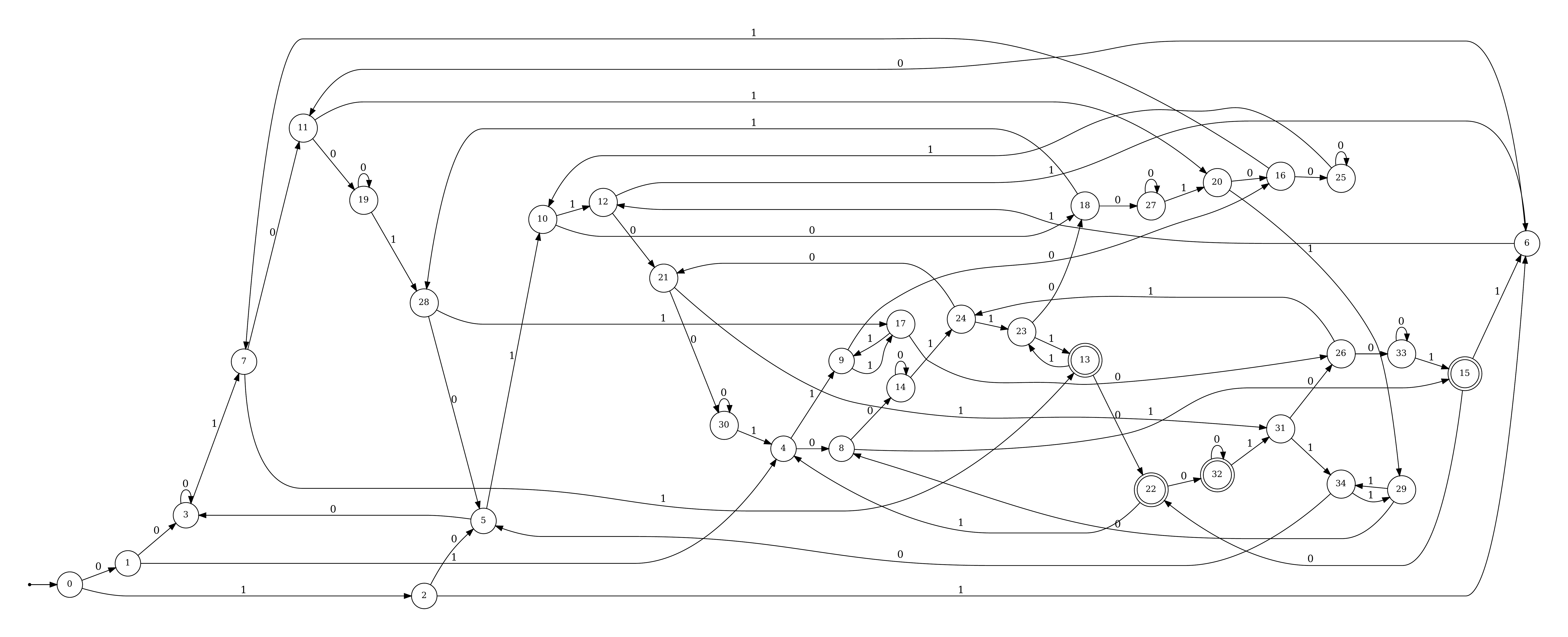}
\end{center}
\caption{LSD-first automaton accepting $n$ if and only if $n!$ is not 
a sum of $3$ squares.}
\label{nfac6}
\end{figure}

\subsection{Overlap-free Dyck words}

In this section we cover another interesting application of transducers, to overlap-free Dyck words.
We need a few definitions.

A word of the form $axaxa$, where $a$ is a single letter, and $x$ is a (possibly empty) string, is called an {\it overlap}.
An example is the French word {\tt entente}.
A string is said to be {\it overlap-free\/}
if it contains no block that is an overlap.
A binary word is a {\it Dyck word} if it represents a string of balanced parentheses, if $1$ is treated as a left paren and $0$ as a right paren.
The {\it nesting level\/} $N(x)$ of a finite Dyck word $x$ is defined as follows: 
$N(\epsilon) = 0$, $N(1x0) = N(x) + 1$ if
$x$ is balanced, and $N(xy) = \max(N(x), N(y))$ if
$x,y$ are both balanced.

In the paper \cite{Mol&Rampersad&Shallit:2023}, the following result is proved:  
there are arbitrarily long overlap-free binary words of nesting level $3$.  Here we can prove the same result using a different, explicit construction involving the Thue-Morse morphism $\mu$, where $\mu(0) = 01$ and
$\mu(1) = 10$:

\begin{theorem}
Define $x_0 = 10$ and $x_{n+1} = \mu(101 \mu(x_n) 101)$ for $n \geq 0$.
Define $y_n = 00x_n 11$ for $n \geq 0$.  Then
$y_n$ is overlap-free for $n \geq 0$ and
$y_n$ has nesting level $3$ for $n \geq 1$.
\end{theorem}

\begin{proof}
We start by defining the infinite binary
word ${\bf d} = 01\, y_0 y_1 y_2 \cdots$.   
We will prove the results about the $y_n$, indirectly, by proving results
about $\bf d$ instead.
We claim that
\begin{itemize}
\item[(a)] $|y_n|= 6 \cdot 4^n$ for $n \geq 0$
\item[(b)] $y_n = {\bf d}[2 \cdot 4^n..2\cdot 4^{n+1}-1]$ for $n \geq 0$
\item[(c)] Define $I = i \in \{2\} \cup \{2\cdot 4^n - 2 \,:\, n\geq 1\} \cup \{2 \cdot 4^n + 1 \,
:\, n \geq 1\}$.  Then $\bf d$ differs from $\mu^2({\bf d})$ precisely
at the indices in $I$, and furthermore $\bf d$ is the unique infinite binary
word starting with $0$ with this property.
\item[(d)] $\bf d$ is generated by the 
$10$-state DFAO in Figure~\ref{fig11}.
\begin{figure}[H]
\begin{center}
\includegraphics[width=5.75in]{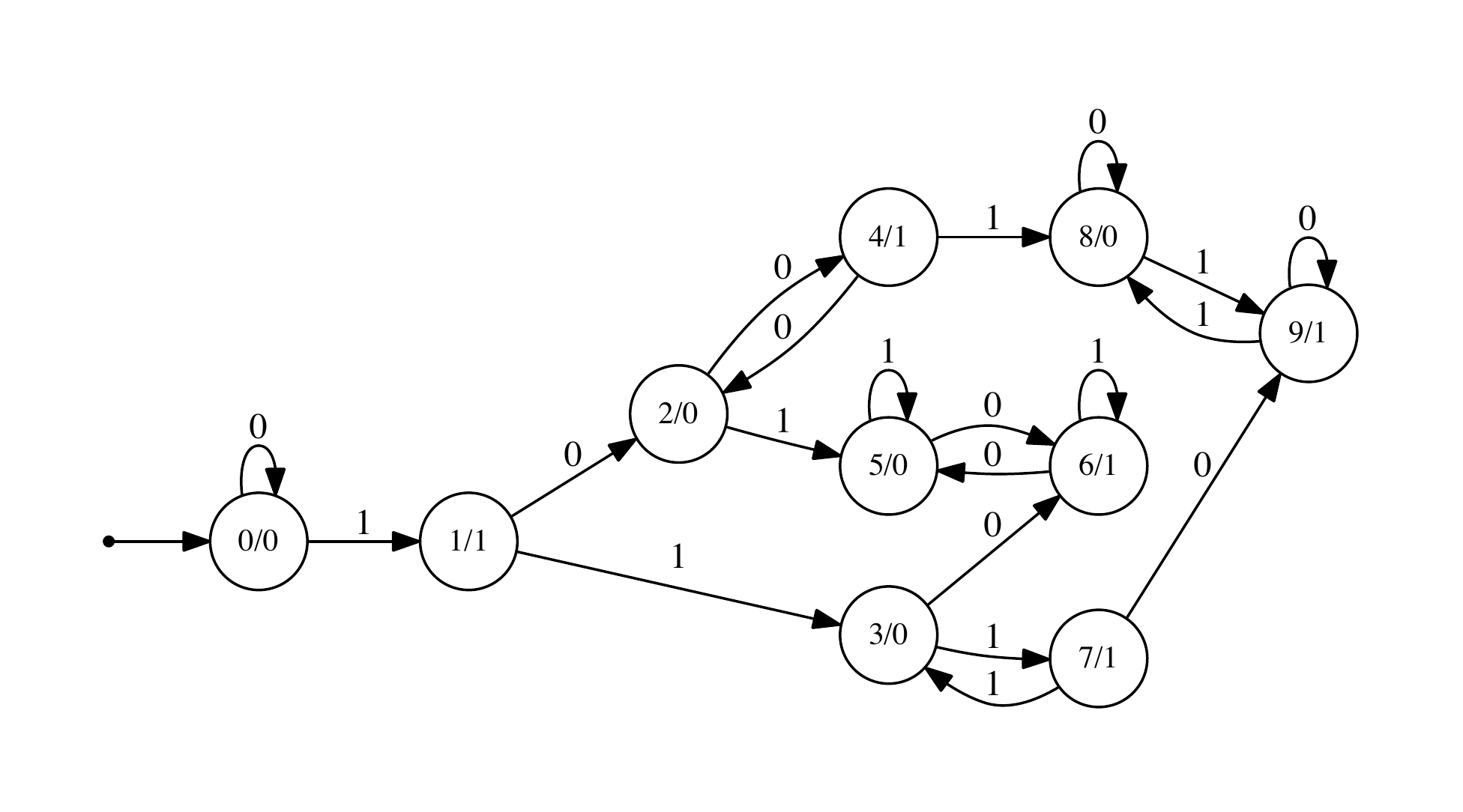}
\end{center}
\caption{DFAO for the sequence $\bf d$.}
\label{fig11}
\end{figure}
\item[(e)] Each $y_n$ is overlap-free, for $n\geq 0$.
\item[(f)] Each $y_n$ has nesting level $3$, for $n \geq 1$.
\end{itemize}

We now prove each of the claims.
\begin{itemize}
\item[(a)]  Clearly $|x_{n+1}| = 4|x_n| + 12$ for $n \geq 0$.
That, together with $|x_0| = 2$, gives $|x_n| = 6 \cdot 4^n - 4$ by
an easy induction.  Hence $|y_n| = 6 \cdot 4^n$.

\item[(b)]  Again, an easy induction shows
$\sum_{0\leq i \leq n} |y_i| = \sum_{0\leq i \leq n} 6 \cdot 4^i = 2 \cdot 4^{n+1} - 2$, which proves the claim.

\item[(c)]   It suffices to show that
${\bf d}[0..2\cdot 4^{i+1} - 1]$ differs from
$\mu^2({\bf d}[0..2\cdot 4^i - 1])$ precisely at
the indices $I \, \cap \, \{0,\ldots, 2\cdot4^{i+1}-1 \}$.
To see this, note that
\begin{align*}
{\bf d}[0..2\cdot 4^{i+1}-1] &= 01\, y_0 \prod_{1 \leq j\leq i} y_j \\
&= (01)(001011) \prod_{1 \leq j \leq i} 00 \, x_j\, 11 \\
&= (01)(001011) \prod_{1 \leq j \leq i} 00100110 \, \mu^2(x_{j-1}) \, 10011011 \\
&= (01)(001011) \prod_{0 \leq j < i} 00100110 \, \mu^2(x_j) \, 10011011.
\end{align*}
On the other hand,
\begin{align*}
\mu^2({\bf d}[0..2\cdot 4^i - 1]) &= \mu^2 (01 \prod_{0\leq j < i} y_j) \\
&= 01101001 \prod_{0 \leq j < i} \mu^2(y_j)  \\
&= (01)(101001) \prod_{0 \leq j < i} \mu^2(00 \, x_j \, 11) \\
&= (01)(101001) \prod_{0 \leq j < i} 01100110 \, \mu^2(x_j) \, 10011001.
\end{align*}
Inspection now shows that 
the differences occur precisely at the indices mentioned.

\item[(d)]  We save the automaton in {\tt Walnut} and then verify
the word generated satisfies the criterion of part (c).
\begin{verbatim}
morphism dd "0->01 1->23 2->45 3->67 4->28 5->65 6->56 7->93 8->89 9->98":
morphism c "0->0 1->1 2->0 3->0 4->1 5->0 6->1 7->1 8->0 9->1":
promote D1 dd:
image D c D1:
# the automaton in the figure for the word d

morphism mu2 "0->0110 1->1001":
image DP mu2 D:
# the automaton for mu^2(d)

reg power4 msd_2 "0*1(00)*":
def differ "D[n]!=DP[n]":
# indices where they differ
eval test "An $differ(n) <=> (Ex $power4(x) & x>=4 & (n=2|n=2*x-2|n=2*x+1))":
# check if criterion satisfied
\end{verbatim}
And {\tt Walnut} returns {\tt TRUE}.

\item[(e)]  
We use the following {\tt Walnut} code.  It asserts that there is 
a segment of $\bf d$ inside a $y_n$ that is an overlap.  
When we run it, {\tt Walnut} returns {\tt FALSE}.
\begin{verbatim}
eval has_overlap "Ex,i,n $power4(x) & i>=2*x & n>=1 & i+2*n<8*x & 
   At,u (t>=i & t<=i+n & u=t+n) => D[t]=D[u]":
\end{verbatim}

\item[(f)]
To check the nesting level, we first make a finite-state transducer that,
on input a word $x$, computes the nesting level if it is $\leq 3$ (and outputs
$4$ if it is $\geq 4$ or $<0$).   It is depicted in Figure~\ref{fig12}.
\begin{figure}[H]
\begin{center}
\includegraphics[width=6in]{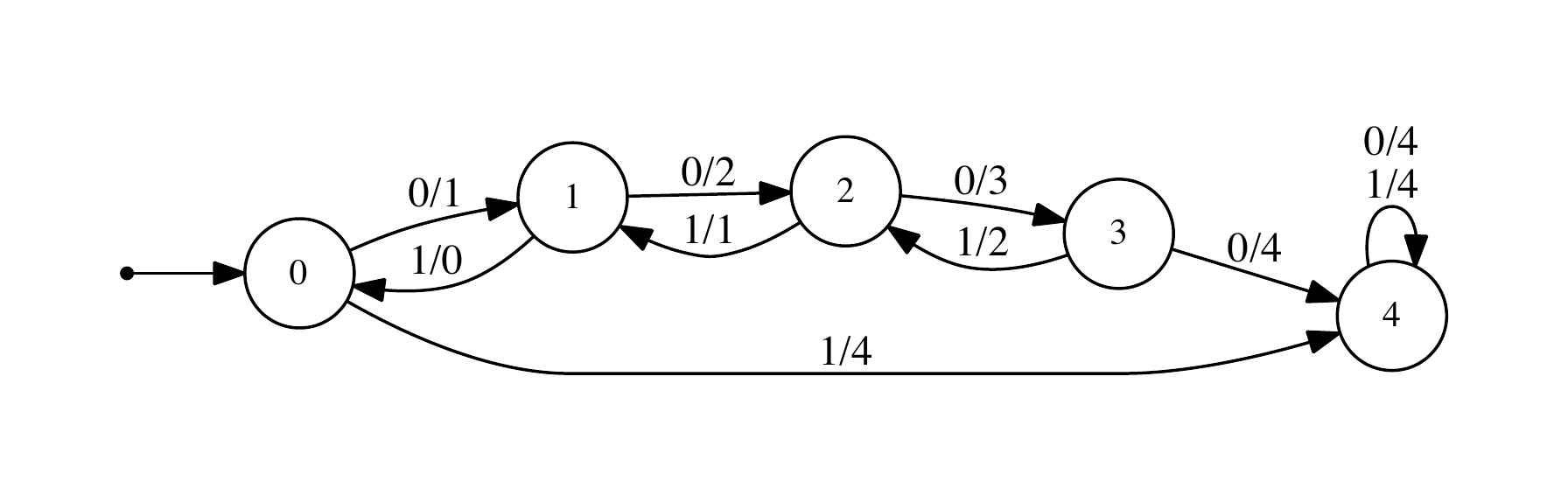}
\end{center}
\caption{Nesting level transducer.}
\label{fig12}
\end{figure}
We now simply run our sequence $\bf d$ through this transducer and
examine the output, as given in
Figure~\ref{fig13}.
\begin{figure}[htb]
\begin{center}
\includegraphics[width=5in]{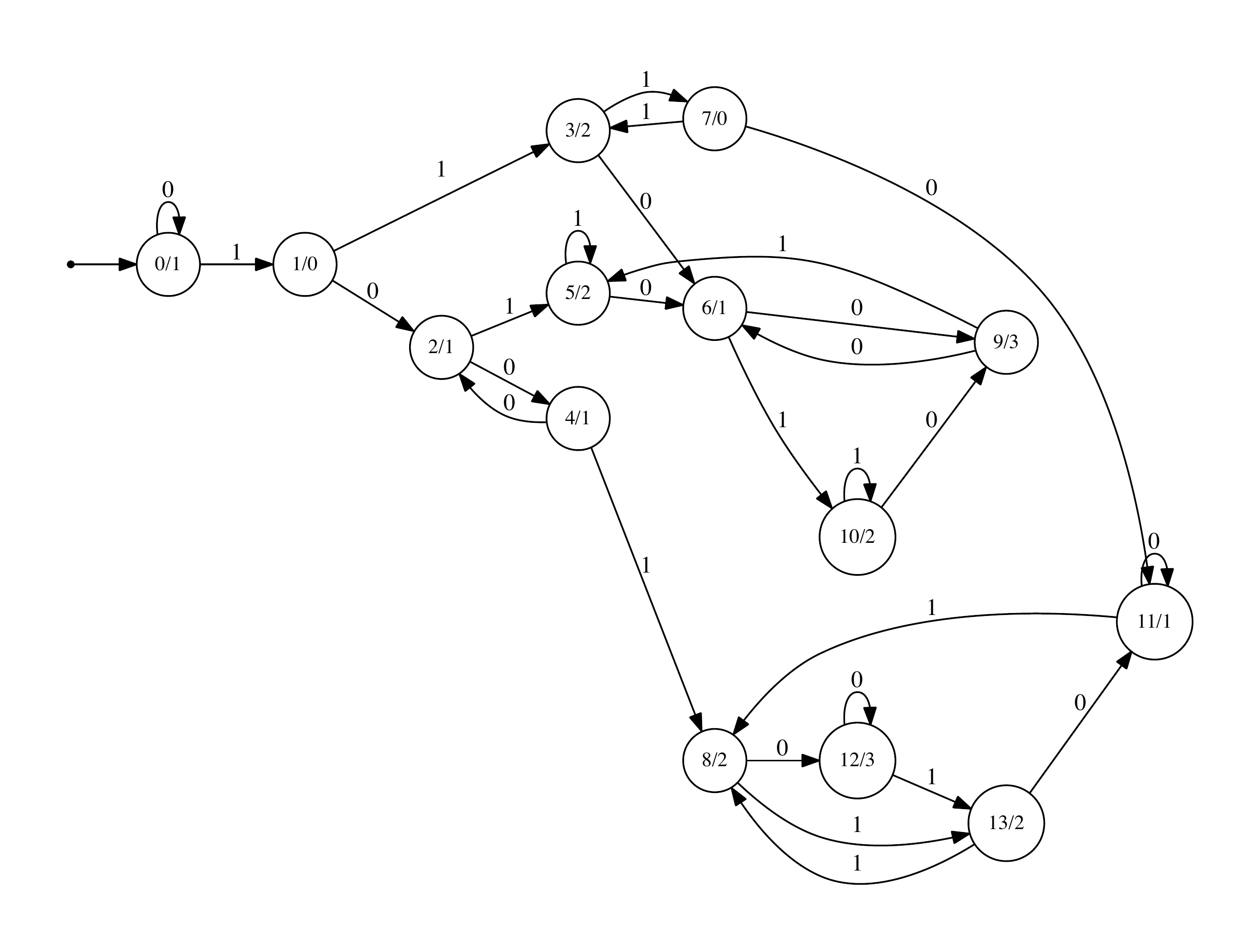}
\end{center}
\caption{Output of transducer applied to $\bf d$.}
\label{fig13}
\end{figure}
It has no outputs of $4$, so the nesting level
of $\bf d$ is $\leq 3$.   We can determine  the positions
of the $0$'s in the transduced sequence as follows:
\begin{verbatim}
transduce DN nest D:
eval zeros "An DN[n]=@0 <=> (Ex $power4(x) & n=2*x-1)":
eval threes "Am,n (m<n & DN[m]=@0 & DN[n]=@0 & m>=7) => 
   Et m<t & t<n & DN[t]=@3":
\end{verbatim}
The first assertion is that $0$'s occur in the transduced sequence
at exactly the positions corresponding to $2\cdot 4^i - 1$.
The second is that between two occurrences of $0$ in the
transduced sequence $\bf dn$ there is always an occurrence of $3$,
provided the first occurrence is at a position $\geq 7$.
{\tt Walnut} returns {\tt TRUE} for both.  Thus each $y_n$, $n \geq 1$,
has nesting level exactly $3$.
\end{itemize}
This completes the proof.
\end{proof}

\subsection{Iterated running sums of Thue-Morse}

Throughout this section, we write $\Sigma_2 = \{0, 1\}$. Recall the iterated running sums of the Thue-Morse sequence as described in Example \ref{exm:tmrunsum}: the transducer that computes the running sum is $T = \inn{V = \{v_0, v_1\}, \Sigma_2, \varphi, v_0, \Sigma_2, \sigma}$, where the transition function \(\varphi \colon V \times \Sigma_2 \to V\) is defined by
$$
\varphi(v_0, 0) = v_0, \quad \varphi(v_0, 1) = v_1, \quad \varphi(v_1, 0) = v_1, \quad \varphi(v_1, 1) = v_0,
$$
and the output function $\sigma: V \times \Sigma_2 \to \Sigma_2$ is defined by 
$$
    \sigma(v_0, 0) = 0, \quad \sigma(v_0, 1) = 1, \quad
    \sigma(v_1, 0) = 1, \quad \sigma(v_1, 1) = 0.
$$
Then we denote the $n$-fold running sum (mod 2) of the Thue-Morse sequence $\mathbf{t}$ by $T^n(\mathbf{t})$, which is defined iteratively: $T^n(\mathbf{t}) = T(T^{n-1}(\mathbf{t}))$. For notational convenience, we write $\mathbf{t}_m = T^m(\mathbf{t})$ to denote the $m$-fold running sum of $\mathbf{t}$, where $\mathbf{t}_0 = \mathbf{t}$.

We can find automata computing the first several running sums in \texttt{Walnut} as follows. Using the \texttt{RUNSUM2} running sum transducer pictured in Figure \ref{fig:RUNSUM2}, we define the automaton \texttt{TSUM1} that computes the first running sum of the Thue-Morse sequence, pictured in Figure \ref{fig:TSUM1}, with the following \texttt{Walnut} command:
\begin{verbatim}
transduce TSUM1 RUNSUM2 T:
\end{verbatim}
\begin{figure}[H]
\begin{center}
    \includegraphics[width=6.5in]{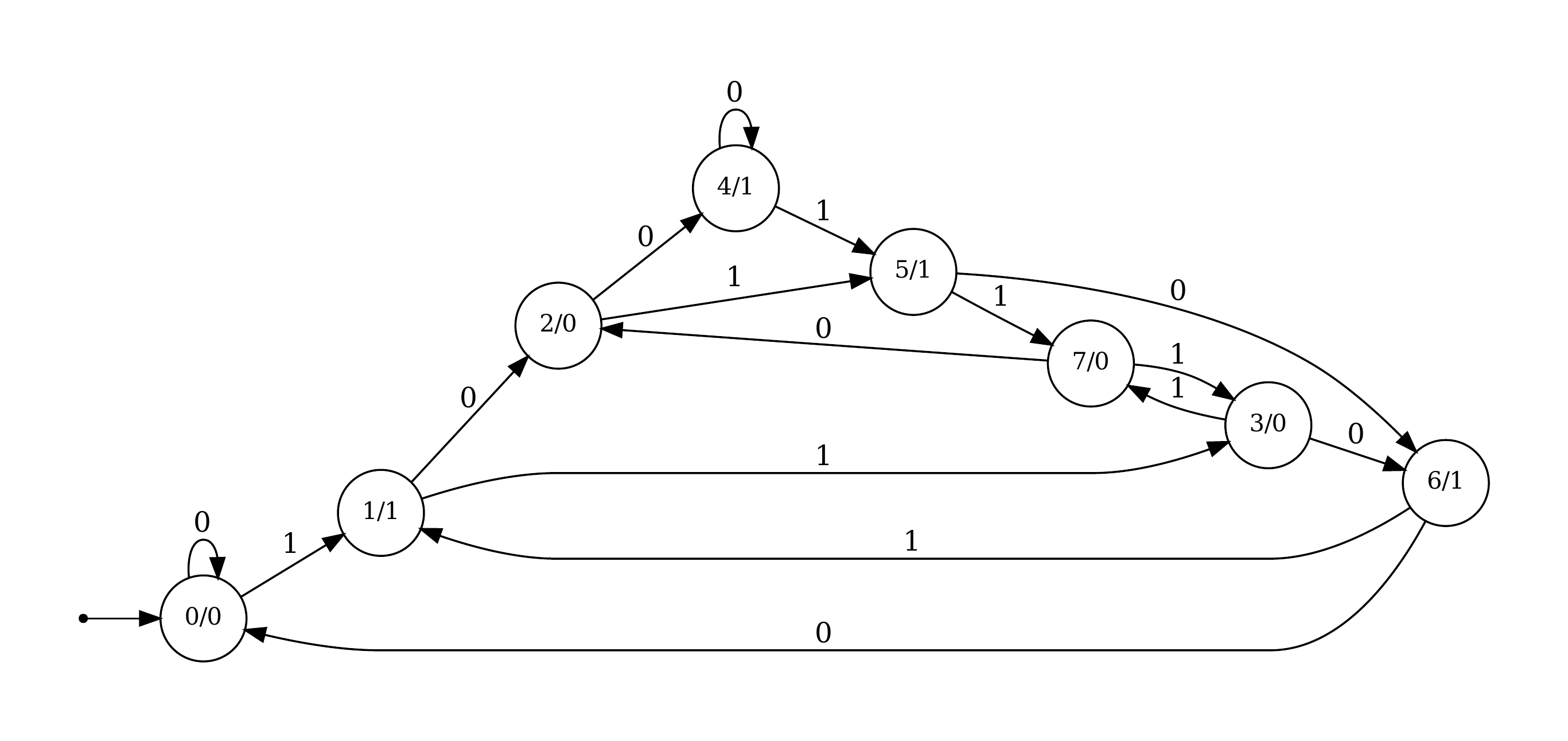}
\end{center}
\caption{DFAO \texttt{TSUM1} computing the running sum of $\mathbf{t}$.} \label{fig:TSUM1}
\end{figure}

To compute the $2$-fold running sum, we apply the \texttt{RUNSUM2} transducer to \texttt{TSUM1} to give the automaton \texttt{TSUM2}, using the following \texttt{Walnut} command:
\begin{verbatim}
transduce TSUM2 RUNSUM2 TSUM1:
\end{verbatim}
\begin{figure}[H]
\begin{center}
    \includegraphics[width=6.5in]{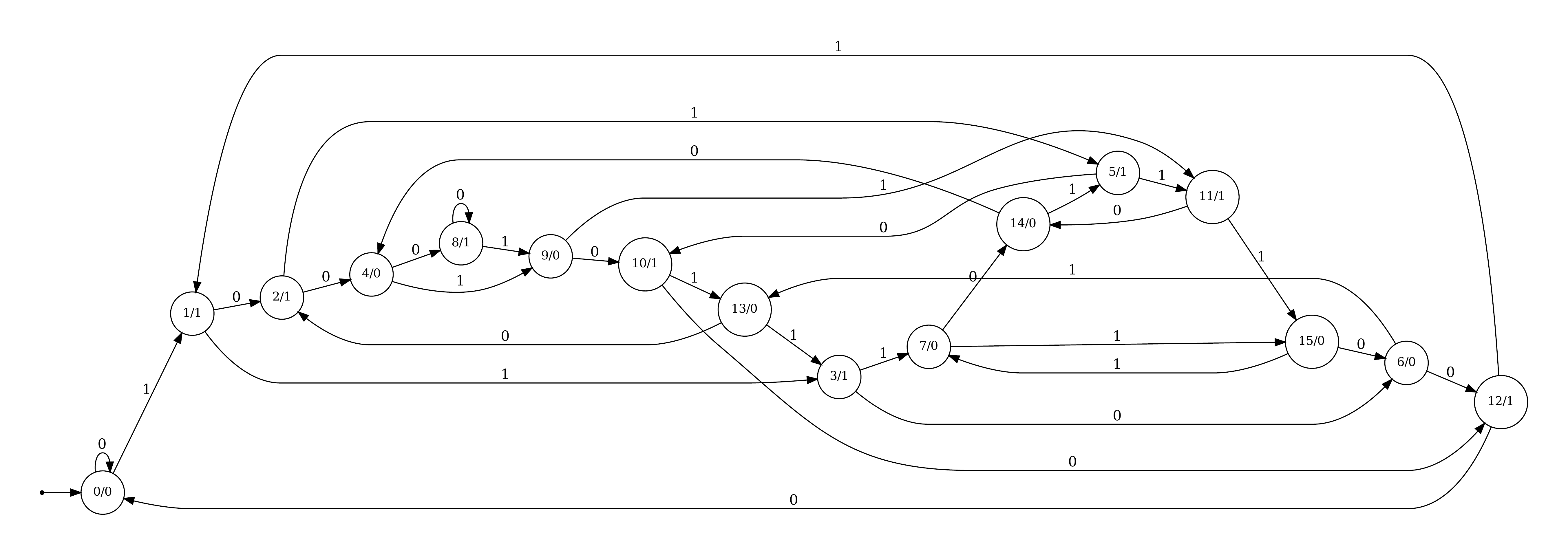}
\end{center}
\caption{DFAO \texttt{TSUM2} computing the $2$-fold running sum of $\mathbf{t}$.} \label{fig:TSUM2}
\end{figure}

We notice that the Thue-Morse automaton has $2$ states, the automaton $\texttt{TSUM1}$ computing the running sum of $\mathbf{t}$ has $8 = 2^3$ states, and the automaton $\texttt{TSUM2}$ has $16 = 2^4$ states. Manually computing the first $34$ running sums as shown above, we get the following sequence of numbers of states of the minimal automata:
\begin{gather*}
    \mathbf{8}, \mathbf{16}, 12, \mathbf{32}, 24, 19, 28, \mathbf{64}, 48, 38, 36, 34, 29, 48, 52, \mathbf{128}, 96, 76, \\
    72, 74, 54, 56, 52, 64, 53, 48, 41, 84, 64, 83, 108, \mathbf{256}, 192, 152, \ldots
\end{gather*}
It is sequence \seqnum{A359228} in the OEIS.
We observe that the minimal automaton computing the $2^n$-fold running sum seems to have $2^{n+3}$ states, which turns out to be true. The remainder of this section is devoted to proving this:

\begin{theorem} \label{thm:tm_min_states}
    For every $n \geq 0$, the minimal DFAO that computes $\mathbf{t}_{2^n}$ has $2^{n+3}$ states.
\end{theorem}



We first begin with a few lemmas that help us characterize the running sums. 

\begin{lemma} \label{lem:tm_binom}
For any $k \geq 0$ and $m \geq 1$,
\begin{equation} \label{eqn:tm_binom}
\mathbf{t}_m[k] \equiv \sum_{j=0}^k \binom{m-1 + (k-j)}{k-j} \mathbf{t}[j] \pmod 2.
\end{equation}
\end{lemma}
\begin{proof}
     Our proof will be by induction on $m$ and $k$. From the definition of the running sum
     \begin{equation} \label{eqn:tm_sum_def}
         \mathbf{t}_m[k] \equiv \sum_{j=0}^k \mathbf{t}_{m-1}[j] \pmod 2,
     \end{equation}
    for all $m \geq 1$ and $k \geq 0$, it is easy to check that
    \begin{equation} \label{eqn:tm_sum_recurrence}
        \mathbf{t}_m[k] \equiv \mathbf{t}_{m-1}[k] + \mathbf{t}_m[k-1] \pmod 2
    \end{equation}
    for any $m \geq 1$ and $k \geq 1$, and
    \begin{equation} \label{eqn:tm_sum_k=0}
        \mathbf{t}_m[0] = \mathbf{t}_{m-1}[0]
    \end{equation}
    for any $m \geq 1$. 
    
    Showing that (\ref{eqn:tm_binom}) holds for all $m \geq 1$ with $k = 0$. is clear from (\ref{eqn:tm_sum_k=0}): clearly (\ref{eqn:tm_binom}) holds for $\mathbf{t}_0[0]$ and by (\ref{eqn:tm_sum_k=0}) we have $\mathbf{t}_m[0] = \mathbf{t}[0]$ for all $m \geq 1$. Showing that (\ref{eqn:tm_binom}) holds for $m = 1$ and all $k \geq 0$ is trivial from (\ref{eqn:tm_sum_def}):
    $$
    \mathbf{t}_1[k] \equiv \sum_{j=0}^k \mathbf{t}[j] \equiv \sum_{j=0}^k \binom{1 - 1 + (k-j)}{k-j} \mathbf{t}[j]  \pmod 2.
    $$
    For the inductive step, suppose that (\ref{eqn:tm_binom}) holds for $\mathbf{t}_{m-1}[k]$ and $\mathbf{t}_m[k-1]$ for $m \geq 2$. From (\ref{eqn:tm_sum_recurrence}), we have
    
    \begin{align*}
        \mathbf{t}_m[k] &\equiv \mathbf{t}_{m-1}[k] + \mathbf{t}_m[k-1]  \pmod 2 \\
        &\equiv \sum_{j=0}^k \binom{m-2 + (k-j)}{k-j} \mathbf{t}[j] + \sum_{j=0}^{k-1} \binom{m-1+(k-1-j)}{k-1-j} \mathbf{t}[j]  \pmod 2 \\
        &\equiv \sum_{j=0}^{k-1} \left[ \binom{m-2 + (k-j)}{k-j} + \binom{m-2+(k-j)}{k-1-j}\right] \mathbf{t}[j] + \binom{m-2}{0} \mathbf{t}[k]  \pmod 2 \\
        &\equiv \sum_{j=0}^{k-1} \binom{m-1+(k-j)}{k-j} \mathbf{t}[j] + \binom{m-1+0}{0} \mathbf{t}[k]  \pmod 2 \\
        &\equiv \sum_{j=0}^k \binom{m-1 + (k-j)}{k-j} \mathbf{t}[j]  \pmod 2,
    \end{align*}
    completing the proof of (\ref{eqn:tm_binom}).
\end{proof}

\begin{lemma} \label{lem:tm_pow2sum}
    For any $k \geq 0$ and $n \geq 0$,
    \begin{equation} \label{eqn:tm_pow2sum}
        \mathbf{t}_{2^n}[k] \equiv \mathbf{t}_1 \left[ \, \left\lfloor \frac{k}{2^n} \, \right\rfloor \right] + \left( \left\lfloor \frac{k}{2^n} \right\rfloor + 1 \right) \mathbf{t}[k \bmod 2^n] \pmod 2
    \end{equation}
\end{lemma}
\begin{proof}
Setting $m = 2^n$ in (\ref{eqn:tm_binom}) gives us
\begin{equation} \label{eqn:tm_binom_pow2}
\mathbf{t}_{2^n}[k] \equiv \sum_{j=0}^k \binom{2^n-1+(k-j)}{k-j} \mathbf{t}[j] \pmod 2
\end{equation}
By Lucas's theorem, $\binom{2^n - 1 + (k-j)}{k-j} \equiv 1 \pmod 2$ if and only if the binary representation of $k-j$ does not contain any $0$'s in its last $n$ bits, or equivalently, $k-j = i 2^n$ for some $i \in \mathbb{N}$, which we may rewrite as $j = k - i 2^n$, which then allows us to write (\ref{eqn:tm_binom_pow2}) as
\begin{equation*}
    \mathbf{t}_{2^n}[k] \equiv \sum_{i=0}^{\lfloor \frac{k}{2^n} \rfloor} \mathbf{t}[k-i 2^n] \pmod 2 
\end{equation*}
As $\mathbf{t}[k-i 2^n]$ is the sum of the bits in the binary representation of $k-i 2^n$, we split it as
\begin{align*}
    \mathbf{t}_{2^n}[k] &\equiv \sum_{i=0}^{\lfloor \frac{k}{2^n} \rfloor} \left( \mathbf{t} \left[ \left\lfloor \frac{k}{2^n} \right\rfloor - i \right] + t[k \bmod 2^n] \right) \pmod 2 \\
    &\equiv \sum_{i=0}^{\lfloor \frac{k}{2^n} \rfloor} \mathbf{t} \left[ i \right] + \left( \left\lfloor \frac{k}{2^n} \right\rfloor + 1 \right) t[k \bmod 2^n] \pmod 2 \\
    &\equiv \mathbf{t}_1 \left[ \, \left\lfloor \frac{k}{2^n} \, \right\rfloor \right] + \left( \left\lfloor \frac{k}{2^n} \right\rfloor + 1 \right) \mathbf{t}[k \bmod 2^n] \pmod 2.
\end{align*}

\end{proof}

From (\ref{eqn:tm_pow2sum}), we get explicit expressions for $\mathbf{t}_{2^n}$:

\begin{corollary} \label{cor:tm_pow2sum_explicit}
For any $n \geq 0$ and $k \geq 0$, if we let $q = \left\lfloor \frac{k}{2^{n+2}} \right\rfloor$ and $r = k \bmod 2^{n+2}$, then
\begin{equation} \label{eqn:tm_pow2sum_explicit}
    \mathbf{t}_{2^n}[2^{n+2}q + r] \equiv
    \begin{cases}
        \mathbf{t}[q] + \mathbf{t}[r] \pmod 2, & \text{if $0 \leq r \leq 2^n - 1$}; \\
        1, & \text{if $2^n \leq r \leq 2 \cdot 2^n - 1$}; \\
        \mathbf{t}[q] + \mathbf{t}[r - 2 \cdot 2^n] \pmod 2, & \text{if $2 \cdot 2^n \leq r \leq 3 \cdot 2^n - 1$}; \\
        0, & \text{if $3 \cdot 2^n \leq r \leq 2^{n+2}-1$}.
    \end{cases}
\end{equation}

\end{corollary}

We can also write this in a more elegant way:

\begin{corollary} \label{lem:tm_sum_as_morphism}
    For every $n \geq 0$, we have 
    $$
        \mathbf{t}_{2^n} = (g_n \circ h^\omega)(0),
    $$
    where $h \colon \Sigma_2^* \to \Sigma_2^*$ is the Thue-Morse morphism defined by $h(0) = 01, h(1) = 10$, and $g_n \colon \Sigma_2^* \to \Sigma_2^*$ is the $2^{n+2}$-uniform morphism defined by 
    $$
        g_n(x) = h^n(x) \, 1^{2^n} \, h^n(x) \, 0^{2^n}
    $$
    for $x \in \Sigma_2$.
\end{corollary}






We now present the main result of this section.

\begin{proof}[Proof of \textbf{Theorem \ref{thm:tm_min_states}}]
The idea of this proof is to use the Myhill-Nerode theorem \cite[Thm.~3.4]{Hopcroft&Ullman:1979} by showing that there are exactly $2^{n+3}$ equivalence classes of the Myhill-Nerode congruence. Define the language
$$L = \{0^* (k)_2 : k \in \mathbb{N}, \mathbf{t}_{2^n}[k] = 1\}$$
to be the set of all binary representations of $k \in \mathbb{N}$ such that $\mathbf{t}_{2^n}[k] = 1$, allowing for arbitrarily many leading zeroes. Define the \textit{Myhill-Nerode congruence} $\sim$ as follows: $x \sim y$ if and only if we have that $xz \in L \iff yz \in L$ for all $z \in \Sigma_2^*$. This is an equivalence relation, and the Myhill-Nerode theorem (Theorem 4.1.8 of \cite{Allouche&Shallit:2003}) gives us that $L$ is a regular language if and only if there are finitely many equivalence classes for $\sim$, and furthermore the number of equivalence classes of $\sim$ is exactly the number of states of a minimal DFA that recognizes $L$. This minimal DFA can then be turned into a minimal DFAO computing $L$ by defining a coding $\lambda \colon Q^* \to \Sigma_2^*$ by $\lambda(q) = 1$ if $q \in F$ and $\lambda(q) = 0$ otherwise for all $q \in Q$.

First note that clearly $(k)_2 \sim 0^m(k)_2$ for any $k, m \geq 0$.
Let $k, k' \in \mathbb{N}$, and let $q = \left\lfloor \frac{k}{2^{n+2}} \right\rfloor$, $r = k \bmod 2^{n+2}$, $q' = \left\lfloor \frac{k'}{2^{n+2}} \right\rfloor$ and $r' = k' \bmod 2^{n+2}$, so $k = 2^{n+2} q + r$ and $k' = 2^{n+2} q' + r'$. 

We  prove that $(k)_2 \sim (k')_2$ if and only if $\mathbf{t}[q] = \mathbf{t}[q']$ and $r = r'$. 

The ``only if'' direction follows immediately from (\ref{eqn:tm_pow2sum_explicit}).

We prove the ``if'' direction contrapositively. Recall that $\mathbf{t}[m]$ is the parity of the $1$s in $(m)_2$ for any $m \in \mathbb N$, and note that $|(r)_2| = |(r')_2| = n+2$. The whole argument will consist of appending strings to both $(k)_2$ and $(k')_2$ to obtain new $k_{\rm new}$ and $k'_{\rm new}$ respectively, and then redefining $r_{\rm new}, r'_{\rm new}, q_{\rm new}$, and $q'_{\rm new}$ accordingly. We then work with $k_{\rm new}$ and $k'_{\rm new}$ instead of $k$ and $k'$. If we show that $(k_{\rm new})_2 \not\sim (k'_{\rm new})_2$, then this will imply $(k)_2 \not\sim (k')_2$ by the definition of the congruence $\sim$.
Here is an example of the appending process for $n=2$:
\begin{gather*}
    (k)_2 = \underbrace{110111}_{q} \underbrace{0101}_{r} \xrightarrow{\text{Append $00$}} \underbrace{11011101}_{q_{\rm new}} \underbrace{01 \mathbf{00}}_{r_{\rm new}} = (k_{\rm new})_2
\end{gather*}

We now proceed with the cases.

\underline{Case 1: $\mathbf{t}[q] \neq \mathbf{t}[q']$ and $r = r'$.} 
We append $0^{n+2}$ to both $(k)_2$ and $(k')_2$ to get $(k_{\rm new})_2 = (k)_2 0^{n+2}$ and $(k'_{\rm new})_2 = (k')_2 0^{n+2}$. This gives that $r_{\rm new} = r'_{\rm new} = 0$ and $\mathbf{t}[q_{\rm new}] \neq \mathbf{t}[q'_{\rm new}]$ (the parity of $(q)_2$ and $(q')_2$ is unaltered, as $(r)_2 = (r')_2$). But then by (\ref{eqn:tm_pow2sum_explicit}) we get $\mathbf{t}_{2^n}[k_{\rm new}] \neq \mathbf{t}_{2^n}[k'_{\rm new}]$, so $(k_{\rm new})_2 \not\sim (k'_{\rm new})_2$.

For the remainder of the proof, we will imply the ``new'' subscript whenever we append anything to $(k)_2$ and $(k')_2$.

\underline{Case 2: $\mathbf{t}[q] = \mathbf{t}[q']$ and $r \neq r'$.}
If $(r)_2$ and $(r')_2$ differ on an odd number of bits, then we can append $0^{n+2}$ to both $(k)_2$ and $(k')_2$ and get $r = r' = 0$ and $\mathbf{t}[q] \neq \mathbf{t}[q']$, so Case 1 applies.

Throughout the remainder of this case, assume $(r)_2$ and $(r')_2$ differ on an even number of bits. By appending $0^m$ for a sufficient $m$, we can further assume that $(r)_2$ and $(r')_2$ differ in exactly two bits, and note that the equality $\mathbf{t}[q] = \mathbf{t}[q']$ still holds, as we carry over an even number of differing bits of $(r)_2$ and $(r')_2$ into $(q)_2$ and $(q')_2$. We can also ensure that $(r)_2$ and $(r')_2$ differ in the first bit by appending some string to both $(k)_2$ and $(k')_2$, and without loss of generality we can assume that $(r)_2$ starts with $0$ and $(r')_2$ starts with $1$. We now consider some subcases.

If $(r)_2$ starts with $01$ and $(r')_2$ starts with $11$ then we have $2^n \leq r \leq 2 \cdot 2^{n}-1$ and $3 \cdot 2^n \leq r' \leq 2^{n+2} - 1$, so $\mathbf{t}_{2^n}[k] = 1 \neq 0 = \mathbf{t}_{2^n}[k']$, so $(k)_2 \not\sim (k')_2$.

If $(r)_2$ starts with $00$ and $(r')_2$ starts with $11$, then by appending $0$ to $(k)_2$ and $(k')_2$ we get that $(r)_2$ starts with $00$, $(r')_2$ starts with $10$, and $\mathbf{t}[q] \neq \mathbf{t}[q']$ (as the leading $1$ in $(r')_2$ flips $\mathbf{t}[q']$).  As the remaining $n$ bits of $(r)_2$ and $(r')_2$ are the same (as we originally assumed that $(r)_2$ and $(r')_2$ differ in two bits, so they now differ in just the first!), we see that $\mathbf{t}[r-2 \cdot 2^{n}] = \mathbf{t}[r']$, and as $0 \leq r \leq 2^{n-1}$ and $2 \cdot 2^n \leq r' \leq 3 \cdot 2^n - 1$ we have $\mathbf{t}_{2^n}[k] = \mathbf{t}[q] + \mathbf{t}[r-2 \cdot 2^{n}] \neq \mathbf{t}[q'] + \mathbf{t}[r'] = \mathbf{t}[k']$, so $(k)_2 \not\sim (k')_2$.

If $n > 0$ and $(r)_2$ starts with $010$ and $(r')_2$ starts with $100$, then by appending $0$ to $(k)_2$ and $(k')_2$ we get that $(r)_2$ starts with $10$ and $(r')_2$ starts with $00$, and $\mathbf{t}[q] \neq \mathbf{t}[q']$, so we can repeat the same argument as in the prior paragraph.

If $n > 0$ and $(r)_2$ starts with $011$ and $(r')_2$ starts with $101$, then we append $0$ to $(k)_2$ and $(k')_2$ and get that $(r)_2$ starts with $11$ and $(r')_2$ starts with $01$, which gives $3 \cdot 2^n \leq r \leq 2^{n+2} - 1$ and $2^n \leq r' \leq 2 \cdot 2^n - 1$, so $\mathbf{t}_{2^n}[k] = 1$ and $\mathbf{t}_{2^n}[k'] = 0$, thus $(k)_2 \not\sim (k')_2$. 

When $n = 0$ and $(r)_2 = 01$ and $(r')_2 = 10$, we append $0$ to $(k)_2$ and $(k')_2$ to get $(r)_2 = 10$ and $(r')_2 = 00$, which gives us again that $\mathbf{t}[q] \neq \mathbf{t}[q']$ and $\mathbf{t}[r-2 \cdot 2^{n}] = \mathbf{t}[r'] = 0$, and as $0 \leq r \leq 2^{n-1}$ and $2 \cdot 2^n \leq r' \leq 3 \cdot 2^n - 1$ we have $\mathbf{t}_{2^n}[k] = \mathbf{t}[q] + \mathbf{t}[r-2 \cdot 2^{n}] \neq \mathbf{t}[q'] + \mathbf{t}[r'] = \mathbf{t}[k']$, so $(k)_2 \not\sim (k')_2$. 

Lastly, if $(r)_2$ starts with $00$ and $(r')_2$ starts with $10$, then as we assumed that $(r)_2$ and $(r')_2$ differ in two bits, the last $n$ bits of $(r)_2$ and $(r')_2$ differ in exactly one bit, hence $\mathbf{t}[r] \neq \mathbf{t}[r' - 2 \cdot 2^{n}]$, and as $0 \leq r \leq 2^n-1$ and $2 \cdot 2^n \leq r' \leq 3 \cdot 2^n - 1$ we have $\mathbf{t}_{2^n}[k] \neq \mathbf{t}_{2^n}[k']$ by (\ref{eqn:tm_pow2sum_explicit}), so $(k)_2 \not\sim (k')_2$.

\underline{Case 3: $t[q] \neq t[q']$ and $r \neq r'$.} If $(r)_2$ and $(r')_2$ differ in more than one bit, we can append $0^m$ for a sufficient $m$ to both $(k)_2$ and $(k')_2$ to move only the first pair of differing bits of $(r)_2$ and $(r')_2$ into $(q)_2$ and $(q')_2$, which will give $\mathbf{t}[q] = \mathbf{t}[q']$ and $r \neq r'$, so Case 2 applies.

If $(r)_2$ and $(r')_2$ differ in one bit, we can append $0^m$ for a sufficient $m$ to both $(k)_2$ and $(k')_2$ to move the differing bits to the front of $(r)_2$ and $(r')_2$, so we can assume without loss of generality that $r$ starts with $0$ and $r'$ starts with $1$. As the second bits must be equal, we can just check both cases. If $r$ starts with $01$ and $r'$ starts with $11$, then clearly $\mathbf{t}_{2^n}[k] = 1$ and $\mathbf{t}_{2^n}[k'] = 0$ by (\ref{eqn:tm_pow2sum_explicit}). Otherwise, if $r$ starts with $00$ and $r'$ starts with $10$, then as all the other bits of $r$ and $r'$ must be equal we have $\mathbf{t}[r] = \mathbf{t}[r' - 2 \cdot 2^n]$, but as $t[q] \neq t[q']$, by (\ref{eqn:tm_pow2sum_explicit}) we get $\mathbf{t}_{2^n}[k] \neq \mathbf{t}_{2^n}[k']$, so $(k)_2 \not\sim (k')_2$ in both cases, completing the proof of Case 3.

These three cases complete the ``if'' direction. Thus, we have established that $(k)_2 \sim (k')_2$ if and only if $\mathbf{t}[q] = \mathbf{t}[q']$ and $r = r'$.

From here, it is clear that $\sim$ has precisely $2^{n+3}$ equivalence classes each represented by an element of $\{0, 1\}^{n+3}$; binary strings of length $n+3$. The first bit represents $\mathbf{t}[q]$, and the remaining $n+2$ bits represent $r$. For any binary string $(k)_2$ longer than $n+3$, it lies in the equivalence class corresponding to $\mathbf{t}[q]$ and $r$, e.g., $10111$ lies in the equivalence class represented by $011$ in the case of $n=0$. 

Clearly, binary strings of length $<n+3$ can be viewed as binary strings of length $n+3$ by prepending sufficiently many zeroes, e.g., $1$ lies in the equivalence class represented by $001$ in the case of $n=0$.
 
\end{proof}

\section{Going beyond \texorpdfstring{$k$}{k}-automatic sequences}

\subsection{Generalizations to non-uniform morphisms} \label{sec:nonuniform}
So far, we have only worked with transducers for $k$-automatic sequences and automata that take input over base-$k$ representations of numbers. 
To generalize this, suppose we have a sequence $(x_n)_{n \geq 0}$ that is automatic in an arbitrary numeration system $\mathcal N$ where the digits of all representations are in $\Sigma = \{0, \ldots, k-1\}$ for some $k \in \mathbb N$, written most significant digit first (least-significant-digit-first numeration systems are considered in \ref{sec:lsdtransduce}). 

The sequence $(x_n)_{n \geq 0}$ is computed by a DFAO $M = \langle Q, \Sigma, \delta, q_0, \Delta, \lambda \rangle$, where for each $q \in Q$ and $a \in \Sigma$, we have $|\delta(q, a)| \in \{0, 1\}$ (that is, there is \textit{at most} one edge going out of $q$ on input $a$). 
When $M$ receives the $\mathcal N$-representation (i.e., the representation in the numeration system $\mathcal N$) of $n$ as input, the automaton computes $x_n$.

\begin{example}[Fibonacci-Thue-Morse] \label{exm:ftm}

    Consider the \textit{Fibonacci-Thue-Morse} sequence $\mathbf{ftm} = (\mathbf{ftm}[n])_{n \geq 0}$, where $\mathbf{ftm}[n]$ is defined to be the sum (mod $2$) of the Fibonacci representation of $n$ (as in \S \ref{sec:numeration}). The first few terms of the sequence are as follows:
$$
\mathbf{ftm} = 01110100100011000101 \cdots
$$
The sequence $\mathbf{ftm}$ is \textit{Fibonacci-automatic}: there is a finite automaton with output $M$ that computes $\mathbf{ftm}[n]$ on input the Fibonacci representation of $n$. Such an automaton $M$ is pictured in Figure~\ref{nfac7}.     
\begin{figure}[H]
\begin{center}
    \includegraphics[width=4in]{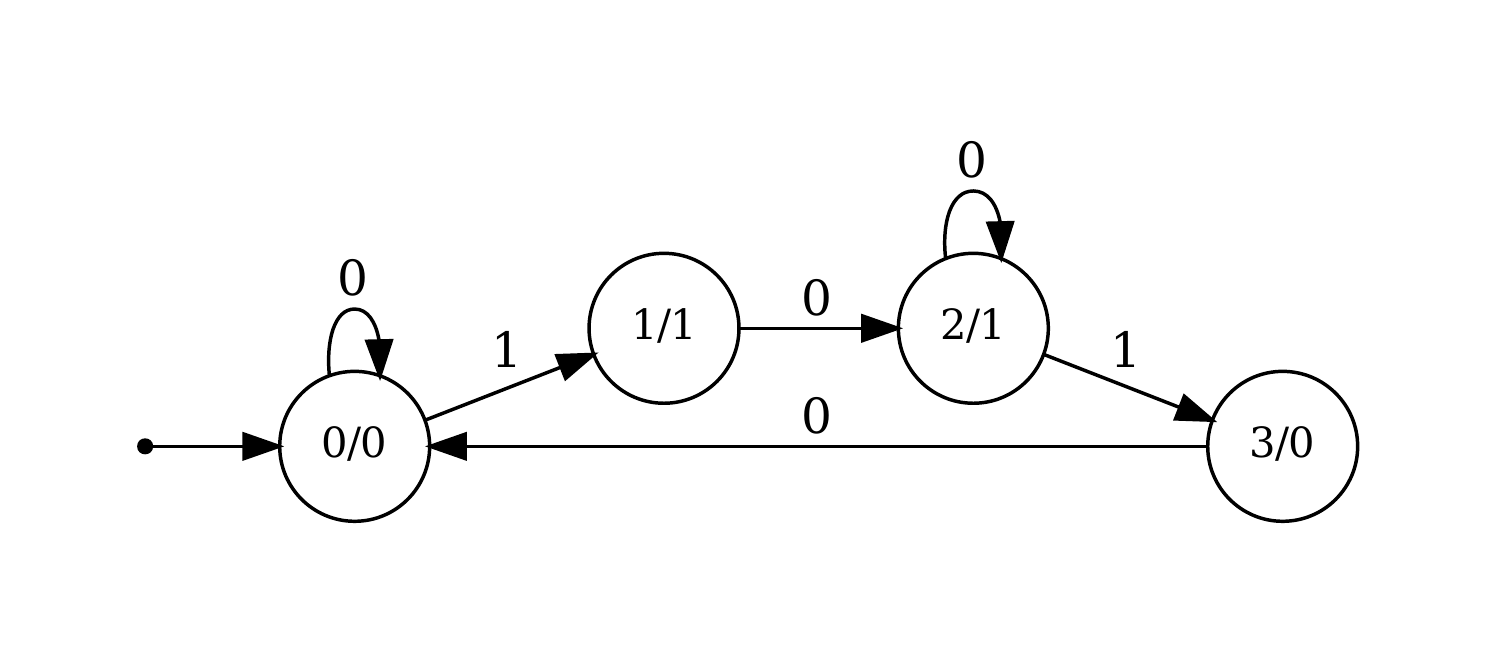}
\end{center}
\caption{Fibonacci automaton $M$ computing $\mathbf{ftm}$.}
\label{nfac7}
\end{figure}
Note that the above automaton is defined on Fibonacci representations only, i.e.,it is only defined on binary strings that do not have consecutive $1$s.
\end{example}

\medskip

As in the example above, the automaton $M$ is only defined on inputs of valid $\mathcal N$-representations. To allow us to apply our existing transduction algorithm to $M$, we wish to extend $M$ to be defined on all base-$k$ representations of numbers. We accomplish this by adding a dead state for all strings that are not $\mathcal N$-representations. The new automaton  behaves identically to the original on all inputs that are valid $\mathcal N$-representations, but ends in the dead state on an input that does not correspond to a valid $\mathcal N$-representation.

In the example of the \texttt{FTM} automaton, this extension is depicted in Figure~\ref{nfac9}.
\begin{figure}[H]
    \begin{center}
        \includegraphics[width=5in]{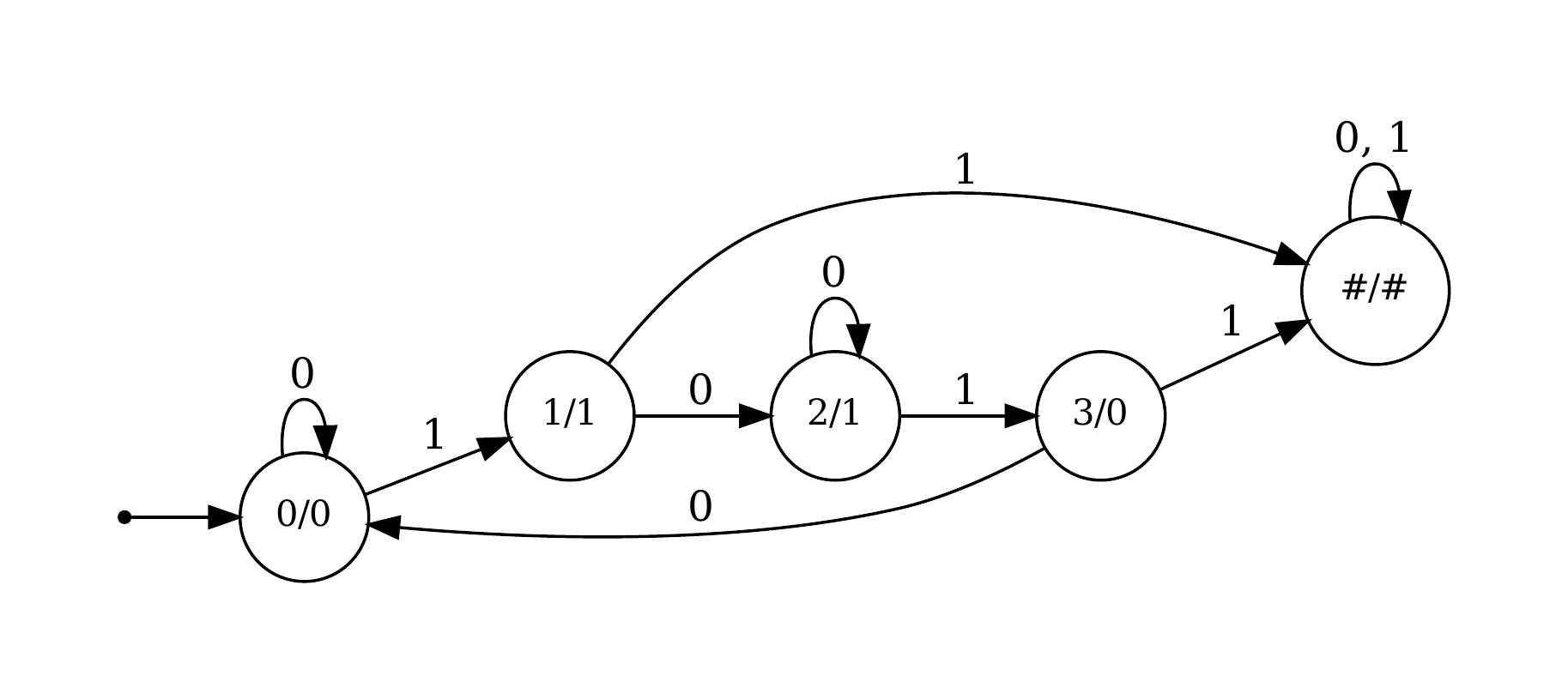}
    \end{center}
    \caption{Extension $M'$ of Fibonacci automaton $M$ to all base-$2$ inputs.}
\label{nfac9}
\end{figure}
Note that all inputs of $1$ from states that were previously undefined on an input of $1$ now go to a dead state with a special output. All valid Fibonacci representations will behave normally, while all binary strings that are not valid Fibonacci representations (i.e., strings with consecutive 1s) will end on the new dead state.

Formally, the extension of $M = \langle Q, \Sigma = \{0, \ldots, k-1\}, \delta, q_0, \Delta, \lambda \rangle$ is a new DFAO $M' = \langle Q', \Sigma, \delta', q_0, \Delta', \lambda' \rangle$, where $Q' = Q \cup \{q_{\#}\}$, $\Delta' = \Delta \cup \{ \# \}$, and $\lambda': Q' \to \Delta'$ is defined as follows:
$$
\lambda'(q') = \begin{cases}
    \#, & \text{if $q' = q_{\#}$}; \\
    \lambda(q'), & \text{otherwise}.
\end{cases}
$$
The transition function $\delta': Q' \times \Sigma \to Q$ is defined as follows:
$$
    \delta'(q', a) = \begin{cases}
        q_{\#}, & \text{if $q' = q_{\#}$ or $\delta(q', a) = \emptyset$}; \\
        q'' \in \delta(q', a), & \text{otherwise, where $\delta(q', a) = \{q''\}$}.
    \end{cases}
$$

Intuitively, $M'$ sends all inputs that are not valid representations to a new dead state $q_{\#}$ with output $\#$. As $M'$ is now defined over all inputs in $\Sigma^*$, it computes a $k$-automatic sequence $y = (y_n)_{n \geq 0}$. We now show that $M'$ behaves identically as $M$ on $\mathcal N$-representations, i.e., both automata compute the same output when given an $\mathcal N$-representation. The idea is for valid representations of numbers in $\mathcal N$ to be instead read in base $k$, resulting in possibly another number. For example, the string $10100$ represents $11$ in Fibonacci representation, but also represents $20$ in base $2$. If the base-$k$ representation of $m \in \mathbb N$ is the same as the $\mathcal N$-representation of some $n \in \mathbb N$, then we want the equality $y_m = x_n$ to hold. This is formalized by Theorem~\ref{trans5}.
\begin{theorem} \label{thm:sequencePadMatches}
    For $m \in \mathbb N$ we have
    $$
    y_m = \begin{cases}
        x_n, & \text{if there exists $n \in \mathbb N$ such that $[n]_{\mathcal N} = [m]_k$}; \\
        \#, & \text{otherwise}.
    \end{cases}
    $$
\label{trans5}
\end{theorem}

\begin{proof}
    Given $m \in \mathbb N$, suppose that there exists $n \in \mathbb N$ such that $[n]_{\mathcal N} = [m]_k$. Then
    \begin{align*}
        x_n &= \lambda(\delta(q_0, [n]_{\mathcal N})) \\
        &= \lambda(\delta'(q_0, [n]_{\mathcal N})) \tag*{(as $\delta(q_0, [n]_{\mathcal N}) \neq \emptyset$)} \\
        &= \lambda'(\delta'(q_0, [n]_{\mathcal N})) \tag*{(as $\delta'(q_0, [n]_{\mathcal N}) \neq q_{\#}$)} \\
        &= \lambda'(\delta'(q_0, [m]_k)) \\
        &= y_m.
    \end{align*}
    Conversely, if $y_m \neq \#$, then by definition of $M'$ we must have that $[m]_k$ is a valid representation in $\mathcal N$ of some $n \in \mathbb N$, i.e., $[n]_{\mathcal N} = [m]_k$. Therefore, if such an integer $n$ does not exist, we must have $y_m = \#$, completing the proof.
\end{proof}

We can generalize this with Lemma~\ref{lem:yIsXWithPad}.
\begin{lemma}
    Let $n, m \in \mathbb{N}$ be such that $[n]_{\mathcal N} = [m]_k$. Then 
    $$
    y_0 \cdots y_m = s_0 \, x_0 \, s_1 \, \cdots \, s_{n-1} \, x_{n-1} \, s_n \, x_n,
    $$ 
    where $s_i \in \{\#\}^*$ for all $i = 0, \ldots, n$.
 \label{lem:yIsXWithPad}
\end{lemma}
\begin{proof}

Let $0 = \alpha_0 < \ldots < \alpha_n = m$ where $\alpha_i \in \{0, \ldots, m\}$ is the value when the representation of $i = 0, \ldots, n$ in $\mathcal N$ is instead read in base-$k$ representation, i.e., $[\alpha_i]_k = [i]_{\mathcal N}$. By Theorem \ref{thm:sequencePadMatches}, we have that $x_i = y_{\alpha_i}$ for all $i = 0, \ldots, n$. For $\gamma \in \{0, \ldots, m\}$ such that  $\gamma \neq \alpha_i$ for all $i = 0, \ldots, n$, we know that $[\gamma]_k$ is not a valid representation in $\mathcal N$, hence $y_\gamma = \#$ by Theorem \ref{thm:sequencePadMatches}. The claim follows.
\end{proof}

Suppose we wish to transduce $M$ through the transducer $T = \langle V, \Delta, \varphi, v_0, \Gamma, \sigma \rangle$. Instead of transducing $M$ directly, we transduce the extension $M'$ through an extension $T'$ of $T$. The extension is defined by $T' := \langle V, \Delta', \varphi', v_0, \Gamma', \sigma' \rangle$, where the input and output alphabets are extended to include $\#$; i.e., $\Delta' = \Delta \cup \{\#\}$ and $\Gamma' = \Gamma \cup \{\#\}$, and $\varphi': V \times \Delta' \to V$ is defined as follows:
$$
\varphi'(v, a) = \begin{cases}
    v, & \text{if $a = \#$}; \\
    \varphi(v, a), & \text{otherwise}.
\end{cases}
$$
The output function $\sigma': V \times \Delta' \to \Gamma'$ is defined as follows:
$$
\sigma'(v, a) = \begin{cases}
    \#, & \text{if $a = \#$}; \\
    \sigma(v, a), & \text{otherwise}.
\end{cases}
$$

As $M'$ now generates a $k$-automatic sequence $y = (y_n)_{n \geq 0}$, we may apply Dekking's method to transduce $M'$ through $T'$, giving a new $k$-automatic sequence $T'(y)$.
\begin{theorem} 
    Let $n, m \in \mathbb{N}$ be such that $[n]_{\mathcal N} = [m]_k$. Then $T'(y)_m = T(x)_n$.
\label{thm:transducedSequencePadMatches}
\end{theorem}
\begin{proof}
    We see that
    \begin{align*}
        T'(y)_m = T(x)_n &\iff \sigma(\varphi(v_0, x_0 \cdots x_{n-1}), x_n) = \sigma'(\varphi'(v_0, y_0 \cdots y_{m-1}), y_m) \\
        &\iff \sigma(\varphi(v_0, x_0 \cdots x_{n-1}), x_n) = \sigma'(\varphi'(v_0, y_0 \cdots y_{m-1}), x_n) \tag*{(by Theorem \ref{thm:sequencePadMatches})} \\
        &\iff \sigma(\varphi(v_0, x_0 \cdots x_{n-1}), x_n) = \sigma(\varphi'(v_0, y_0 \cdots y_{m-1}), x_n) \tag*{(as $x_n \neq \#$)} \\
        &\iff \sigma(\varphi(v_0, x_0 \cdots x_{n-1}), x_n) = \sigma(\varphi'(v_0, s_0 x_0 s_1 \cdots s_{n-1} x_{n-1} s_n), x_n)
    \end{align*}
    by Lemma \ref{lem:yIsXWithPad}, where $s_i \in \{\#\}^*$ for all $i = 0, \ldots, n$.
    Therefore, it suffices to show that  $\varphi'(v_0, s_0 \, x_0 \, s_1 \, \cdots \, s_{n-1} \, x_{n-1} \, s_n) = \varphi(v_0, x_0 \cdots x_{n-1})$, but this follows from the definition of $\varphi'$, completing the proof.
\end{proof}

Thus we get

\begin{corollary}
If ${\bf a} = (a_n)_{n \geq 0}$ is a generalized automatic sequence 
over a numeration system $\cal N$, and $T$ is
a deterministic $1$-uniform finite state transducer, then
$T({\bf a})$ is also a generalized automatic sequence over $\cal N$.
\end{corollary}

\begin{example}[XOR of Fibonacci-Thue-Morse]
Recall the Fibonacci-Thue-Morse sequence $\mathbf{ftm}$ and the automaton $M$ that computes it from Example \ref{exm:ftm}.
Consider the transducer $T$ that computes the XOR of consecutive bits (with the first bit outputted always being $0$), pictured in Figure~\ref{nfac8}.
\begin{figure}[H]
    \begin{center}
    \includegraphics[width=4in]{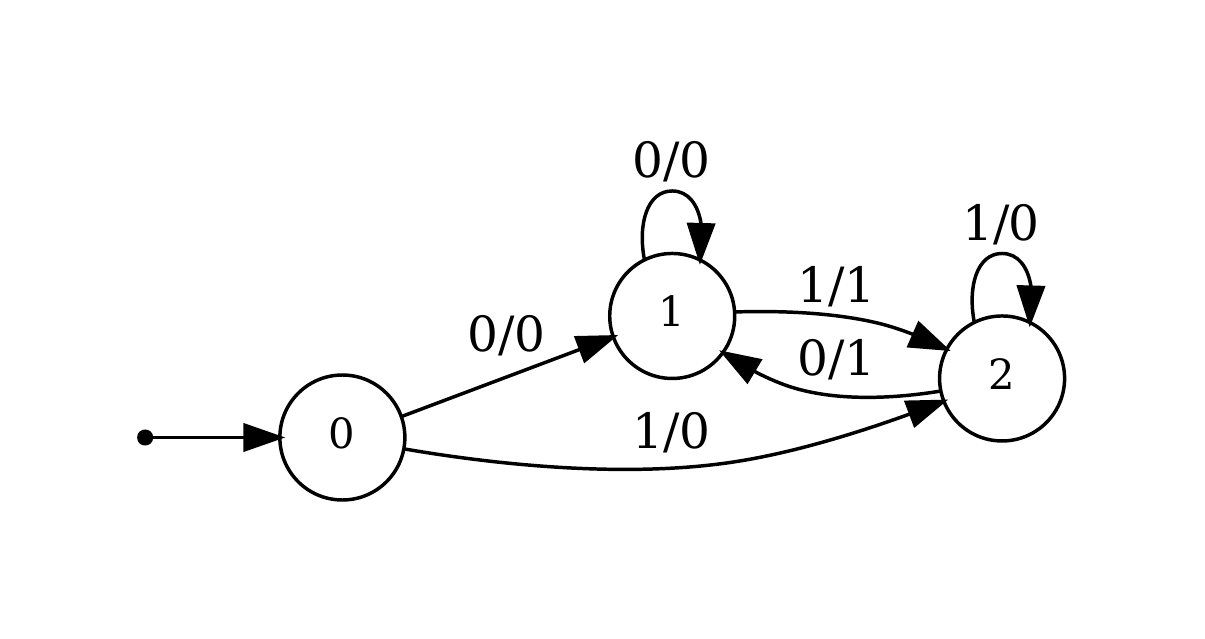}
    \end{center}
    \caption{Transducer $T$ computing the XOR of consecutive bits.}
    \label{nfac8}
\end{figure}

By directly computing from the definition in \S \ref{sec:transducers}, we get that
$$
T(\mathbf{ftm}) = 01001110110010100111 \cdots  .
$$

Our goal is to obtain a Fibonacci automaton that computes $T(\mathbf{ftm})$. However, we cannot directly apply our existing method to $M$ and $T$, as $M$ is not deterministic. Thus, we extend $M$ to a new automaton $M'$ that accepts all base-$2$ inputs, as shown in Figure~\ref{nfac9}.

We also extend the transducer $T$ to the transducer $T'$ as outlined earlier:
\begin{figure}[H]
    \begin{center}
        \includegraphics[width=4in]{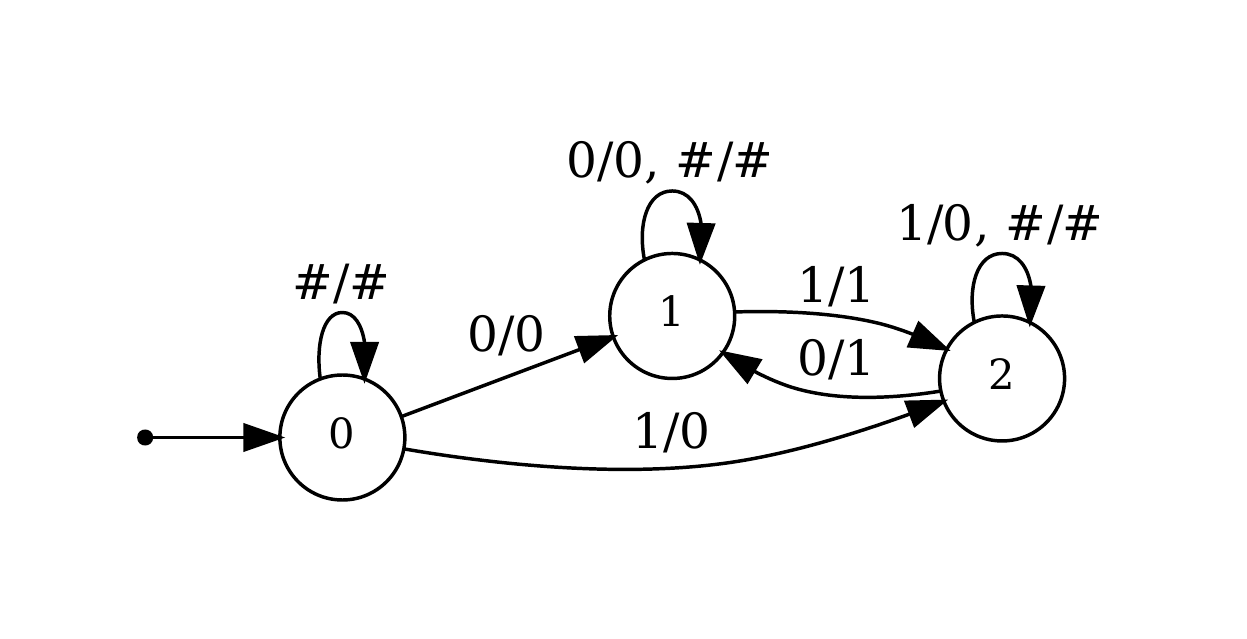}
    \end{center}
    \caption{Extension of the transducer $T$}
\end{figure}

Let $\mathbf{ftm}'$ be the sequence generated by $M'$. As $\mathbf{ftm}'$ is a $2$-automatic sequence, we can apply Dekking's algorithm to transduce $\mathbf{ftm}'$ with $T'$. We show the first few terms of the sequences in the following table. Entries that are left blank correspond to binary inputs that $M$ is not defined on.
\begin{center}
\begin{tabular}{ | c | c | c | c | c | c |} 
$[m]_2$ & $n : [m]_2 = [n]_{\text{fib}}$ & $\mathbf{ftm}[n]$ & $\mathbf{ftm}'[m]$ & $T(\mathbf{ftm})[n]$ & $T'(\mathbf{ftm}')[m]$ \\
\hline
$0$ & $0$ & $0$ & $0$ & $0$ & $0$\\
$1$ & $1$ & $1$ & $1$ & $1$ & $1$\\
$10$ & $2$ & $1$ & $1$ & $0$ & $0$\\
$11$ & & & $\#$ & & $\#$ \\
$100$ & $3$ & $1$ & $1$ & $0$ & $0$\\
$101$ & $4$ & $0$ & $0$ & $1$ & $1$\\
$110$ & & & $\#$ & & $\#$\\
$111$ & & & $\#$ & & $\#$\\
$1000$ & $5$ & $1$ & $1$ & $1$ & $1$ \\
$1001$ & $6$ & $0$ & $0$ & $1$ & $1$ \\
$1010$ & $7$ & $0$ & $0$ & $0$ & $0$ \\
$1011$ &  & & $\#$ & & $\#$ \\
$1100$ &  & & $\#$ & & $\#$ \\
$1101$ &  & & $\#$ & & $\#$ \\
$1110$ &  & & $\#$ & & $\#$ \\
$1111$ &  & & $\#$ & & $\#$ \\
$10000$ & $8$ & $1$ & $1$ & $1$ & $1$ \\
\end{tabular}
\end{center}
As we see in the table, the columns corresponding to $\mathbf{ftm}[n]$ and $\mathbf{ftm}'[m]$ as well as $T(\mathbf{ftm})[n]$ and $T'(\mathbf{ftm}')[m]$ are practically the same, except the places where $T(\mathbf{ftm})$ is undefined is now indicated by an output of $\#$ for $T'(\mathbf{ftm}')$. To recover the original numeration system, we can remove all states with an output of $\#$ in the resulting automaton that computes $T'(\mathbf{ftm}')$. 

\texttt{Walnut} automatically computes the necessary extensions of automata and transducers under the hood, allowing for the same transduction syntax as always. The \texttt{FTM} automaton can be transduced using the \texttt{XOR} transducer (\texttt{Walnut} definition in the Appendix) to give the \texttt{FTMXOR} automaton, pictured in Figure \ref{fig:FTMXOR}, using the \texttt{Walnut} command
\begin{verbatim}
transduce FTMXOR XOR FTM:
\end{verbatim}
\begin{figure}[H]
\begin{center}
    \includegraphics[width=5.5in]{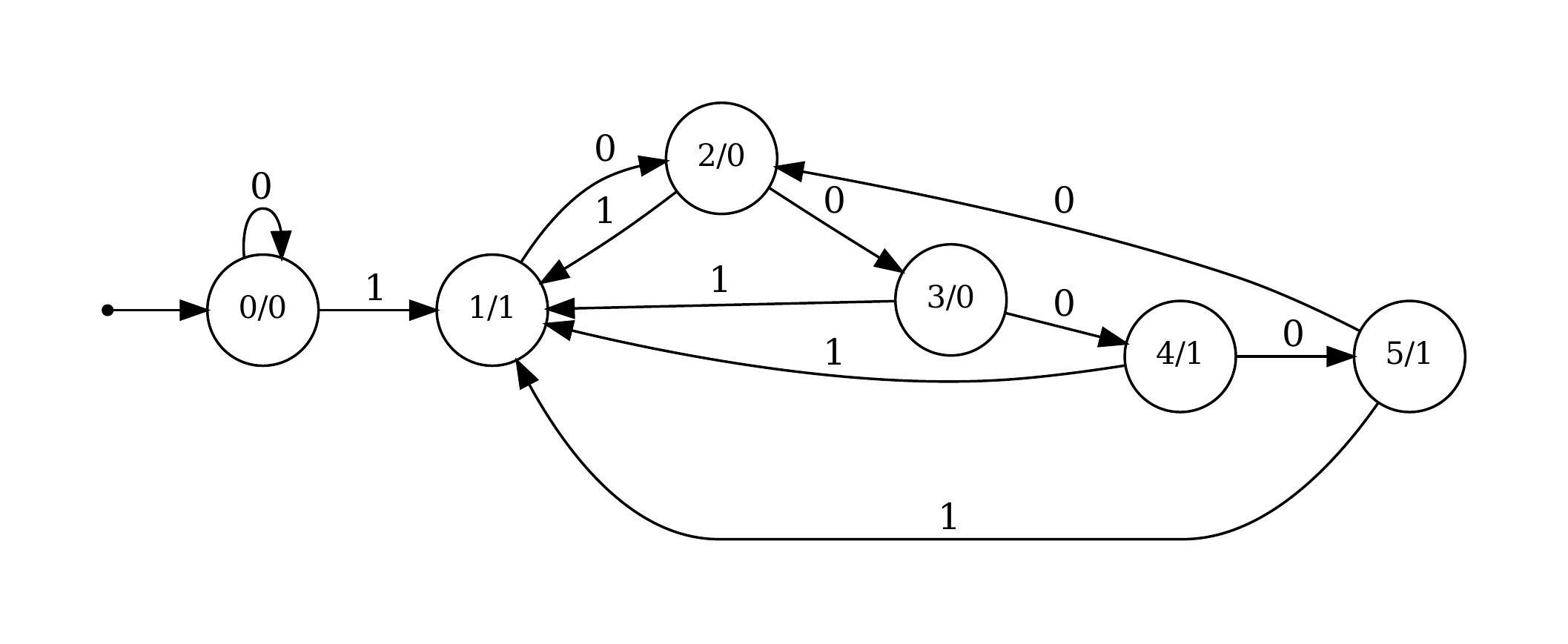}
\end{center}
\caption{DFAO \texttt{FTMXOR} computing $T(\mathbf{ftm})$.}
\label{fig:FTMXOR}
\end{figure}

\end{example}

\subsection{Another example}

Let us look at another example.   Define
$s(n)$ to be the second-to-last bit of the
Fibonacci representation of $n$ (most-significant-digit first), and let $s'(n) = 1-s(n)$.  Table~\ref{slzeck} gives
the first few terms of these sequences:
\begin{table}[H]
\begin{center}
\begin{tabular}{c|ccccccccccccccccc}
$n$ & 0 & 1 & 2 & 3 & 4 & 5 & 6 & 7 & 8 & 9 & 10 & 11 & 12 & 13 & 14 & 15 & 16 \\
\hline
$s(n)$ & 0 & 0 & 1 & 0 & 0 & 0 & 0 & 1 & 0 & 0 & 1 & 0 & 0 & 0 & 0 & 1 & 0 \\
$s'(n)$ & 1& 1& 0& 1& 1& 1& 1& 0& 1& 1& 0& 1& 1& 1& 1& 0& 1
\end{tabular}
\end{center}
\caption{Second-to-last bit of Fibonacci representation.}
\label{slzeck}
\end{table}
Then $s(n)$ is sequence \seqnum{A123740} in the OEIS and $s'(n)$ is sequence \seqnum{A005713}.

Another characterization of $s(n)$, given in OEIS sequence \seqnum{A123740}, is that  
$$s(n) = \lfloor (n+2) \tau \rfloor - \lfloor n\tau \rfloor - 3,$$ where $\tau = (1+\sqrt{5})/2$.  We can
verify this claim using a synchronized automaton for
$\lfloor n \tau \rfloor$, given in \cite[p.~278]{Shallit:2022}, as follows:
\begin{verbatim}
reg sldf msd_fib "0*(10|0)*10":
# second lowest bit of Fibonacci expansion is 1

combine SLDF sldf:
reg shift {0,1} {0,1} "([0,0]|[0,1][1,1]*[1,0])*":
def phin "?msd_fib (s=0 & n=0) | Ex $shift(n-1,x) & s=x+1":
def aseq "?msd_fib Ex,y x=y+3 & $phin(n+3,x) & $phin(n+1,y)":
def test1 "?msd_fib An $aseq(n) <=> SLDF[n]=@0":
\end{verbatim}
and {\tt Walnut} returns {\tt TRUE}.

We now prove the following result:
\begin{theorem}
Let $\bf f$ denote the infinite Fibonacci word \cite{Berstel:1986b}.
Then ${\bf f}[n+1] = (\sum_{0 \leq i \leq n} s'(n)) \bmod 2$ for all $n \geq 0$.
\end{theorem}

\begin{proof}
We use the new capabilities of {\tt Walnut}.  
\begin{verbatim}
def nsldf "?msd_fib ~$sldf(n)":
combine NSLDF nsldf:
# sequence A005713

transduce TS RUNSUM2 NSLDF:
# A005614; this is infinite Fibonacci word shifted by 1

eval test2 "?msd_fib An TS[n]=F[n+1]":
# Walnut returns TRUE
\end{verbatim}
\end{proof}

\subsection{Transducing lsd-first automatic sequences} \label{sec:lsdtransduce}

In this section, we consider the transduction of automata that read in least-significant-digit-first representations to compute elements of an automatic sequence. We say that a DFAO $M$ computing an automatic sequence $\mathbf{x} = (x_n)_{n \geq 0}$ is \texttt{lsd\_k} if it computes computes $x_n$ when it reads the lsd-$k$ representation of $n$. 

Clearly, a DFAO $M$ reads in a most-significant-digit-first representation $d_t d_{t-1} \cdots d_1 d_0$ of $n$ and computes $x_n$, then the reversal $M^R$ of $M$ will compute $x_n$ when it reads the least-significant-digit-first representation $d_0 d_1 \cdots d_{t-1} d_t$ of $n$, and vice versa. This is formalized by a theorem of Allouche and Shallit \cite{Allouche&Shallit:2003}:

\begin{theorem}[4.3.3-4.3.4 of Allouche and Shallit \cite{Allouche&Shallit:2003}]
    If a DFAO $M = \inn{Q, \Sigma, \delta, q_0, \Delta, \lambda}$ with $n$ states computes a function $f \colon \Sigma^* \to \Delta$, then there is a DFAO $M^R$ with at most $|\Delta|^n$ states that computes $f^R \colon \Sigma^* \to \Delta$ defined by $f^R(w) = f(w^R)$.
\end{theorem}

Suppose that we wish to transduce an lsd-$k$ DFAO $M$ that computes an automatic sequence $\mathbf{x} = (x_n)_{n \geq 0}$ with a transducer $T$. To leverage our existing transduction algorithms for msd-$k$ automata, we can instead transduce the msd-$k$ reversal $M^R$, giving us a new msd-$k$ automaton $N^R$ that computes $T(\mathbf{x})$. Reversing $N^R$ then gives us an lsd-$k$ automaton $N$ that also computes $T(\mathbf{x})$.

\begin{example}[Running sums of Thue-Morse in lsd]
Consider the lsd-$2$ automaton that computes the running sum of the Thue-Morse sequence, pictured in Figure \ref{fig:TSUM1_REV}.
\begin{figure}[H]
\begin{center}
    \includegraphics[width=4in]{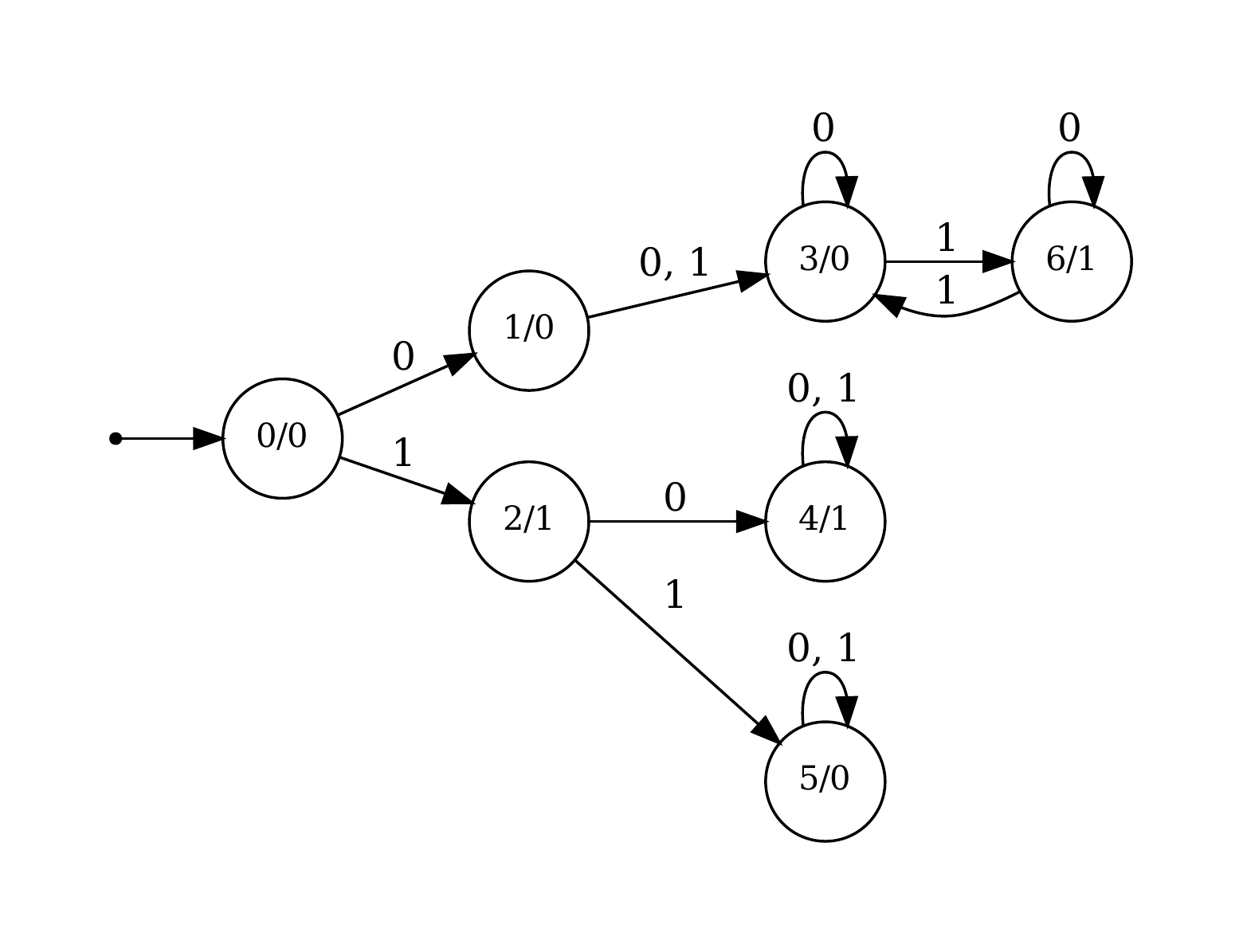}
\end{center}
\caption{An lsd-$2$ automaton \texttt{TSUM1\_REV} computing $\mathbf{t}_1$.}
\label{fig:TSUM1_REV}
\end{figure}
We can check using \texttt{Walnut} that the reversal of \texttt{TSUM1\_REV} is indeed the same one as in Figure \ref{fig:TSUM1} using the \texttt{reverse} command:
\begin{verbatim}
reverse TSUM1_ORIG TSUM1_REV:
eval test "An TSUM1_ORIG[n]=TSUM1[n]":
\end{verbatim}
and {\tt Walnut} returns {\tt TRUE}.
We can then transduce \texttt{TSUM1\_REV} using our \texttt{RUNSUM2} transducer as in Figure \ref{fig:RUNSUM2} to give the automaton in Figure \ref{fig:TSUM2_REV}:
\begin{verbatim}
transduce TSUM2_REV RUNSUM2 TSUM1_REV:
\end{verbatim}
\begin{figure}[H]
\begin{center}
    \includegraphics[width=5in]{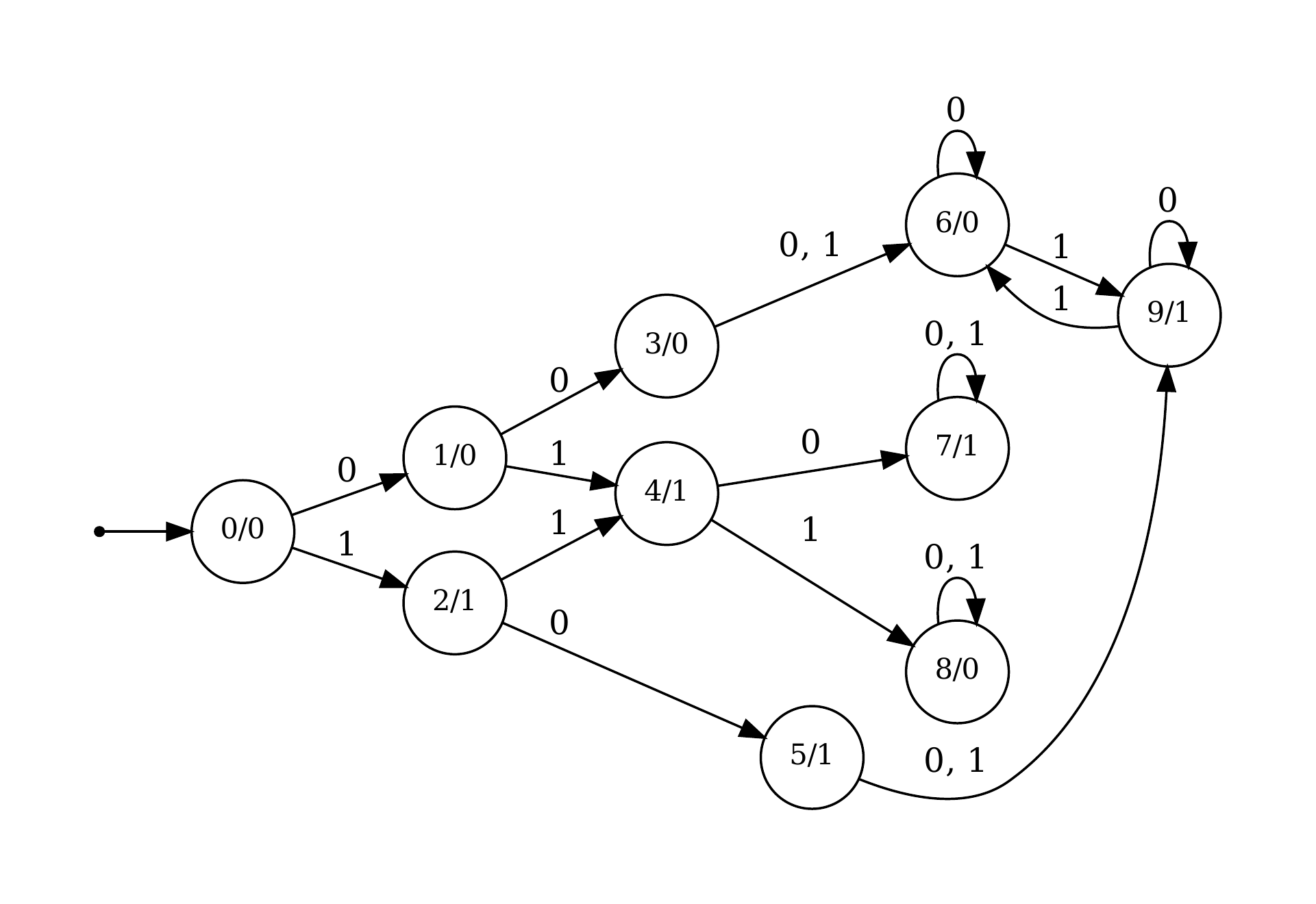}
\end{center}
\caption{An lsd-$2$ automaton \texttt{TSUM2\_REV} computing $\mathbf{t}_2$.}
\label{fig:TSUM2_REV}
\end{figure}
Again, we check that the reversal of \texttt{TSUM2\_REV} is just \texttt{TSUM2} as in Figure \ref{fig:TSUM2}:
\begin{verbatim}
reverse TSUM2_ORIG TSUM2_REV:
eval test2 "An TSUM2_ORIG[n]=TSUM2[n]":
\end{verbatim}
and {\tt Walnut} returns {\tt TRUE}.
\end{example}

Note that the reversed automaton $M^R$ has at most $|\Delta|^{|Q|}$ states even after minimization, so the state complexity of transducing an lsd-$k$ automaton is exponential in the worst case.

If $M$ computes a sequence $(x_n)_{n \geq 0}$ over an arbitrary numeration system $\mathcal{N}$ (for example, the Fibonacci numeration system) taking digits over $\{0, \ldots, k-1\}$, the preceding discussion can be applied to the extension $M'$ of $M$ (as defined in Section \ref{sec:nonuniform}), where the extension $M'$ now computes a $k$-automatic sequence.

\section*{Acknowledgments}

We are grateful to Narad Rampersad for suggesting the example in
Example~\ref{exam1}.

AZ thanks Andrey Boris Khesin for useful discussions, and for providing the statement of Lemma \ref{lem:tm_binom} and the statement and proof of Lemma \ref{lem:tm_pow2sum}.

\newpage

\section*{Appendix}

We provide the \texttt{Walnut} definitions for the transducers that appear in the paper. These are to be defined in \texttt{.txt} files in the \texttt{Transducer Library/} directory of \texttt{Walnut}.

\subsubsection*{\texttt{\# RUNSUM2.txt}}

{\footnotesize \begin{verbatim}
{0, 1}

0
0 -> 0 / 0
1 -> 1 / 1

1
0 -> 1 / 1
1 -> 0 / 0
\end{verbatim}}

\subsubsection*{\texttt{\# RUNPROD1357.txt}}

{\footnotesize \begin{verbatim}
{1, 3, 5, 7}

0
1 -> 0 / 1
3 -> 1 / 3
5 -> 2 / 5
7 -> 3 / 7

1
1 -> 1 / 3
3 -> 0 / 1
5 -> 3 / 7
7 -> 2 / 5

2
1 -> 2 / 5
3 -> 3 / 7
5 -> 0 / 1
7 -> 1 / 3

3
1 -> 3 / 7
3 -> 2 / 5
5 -> 1 / 3
7 -> 0 / 1
\end{verbatim} }

\subsubsection*{\texttt{\# XOR.txt}}

{\footnotesize \begin{verbatim}
{0, 1}

0
0 -> 1 / 0
1 -> 2 / 0

1
0 -> 1 / 0
1 -> 2 / 1

2
0 -> 1 / 1
1 -> 2 / 0
\end{verbatim}}

\subsubsection*{\texttt{\# TSUM1\_REV.txt}}

{\footnotesize \begin{verbatim}
lsd_2

0 0
0 -> 1
1 -> 2

1 0
0 -> 3
1 -> 3

2 1
0 -> 4
1 -> 5

3 0
0 -> 3
1 -> 6

4 1
0 -> 4
1 -> 4

5 0
0 -> 5
1 -> 5

6 1
0 -> 6
1 -> 3
\end{verbatim}}

\end{document}